\documentclass[11pt]{article}

\usepackage{fullpage}
\usepackage{amsmath}
\usepackage{amssymb}
\usepackage{amsthm}
\usepackage{bbm}
\usepackage{algorithm}
\usepackage{algorithmic}
\usepackage{xcolor}
\usepackage{soul}
\usepackage{xfrac}
\usepackage{dsfont}
\usepackage{paralist}
\usepackage{graphicx}
\usepackage{tikz}
\usepackage{subcaption}

\usetikzlibrary{calc}
\usepackage{enumitem}
\usepackage{thmtools,thm-restate}
\usepackage{mathabx}
\usepackage{sansmath}
\usepackage{everyhook} 
\usepackage[colorlinks,urlcolor=blue,citecolor=blue,linkcolor=blue]{hyperref}
\usepackage{graphicx,amsfonts,amsmath,amssymb,epsfig,color,multicol,pstricks,tikz, pgf}
\usetikzlibrary{decorations.markings}
\usetikzlibrary{patterns}
\usetikzlibrary{calc, intersections}
\usepackage{pgf,tikz}
\usetikzlibrary{calc}
\usepackage{tkz-euclide}
\usepackage{pgfplots}
\usepackage{mathabx}
\usepackage{graphicx,amsfonts,amsmath,amssymb,epsfig,color,multicol,pstricks,tikz, pgf}

\restylefloat{figure}

\setlist[enumerate]{ topsep=2pt, parsep=.5mm}

\newcommand\withcomments{1}

\newtheorem{thm}{Theorem}[section]
\newtheorem{lem}{Lemma}[section]
\newtheorem{dfn}[lem]{Definition}
\newtheorem{clm}[lem]{Claim}
\newtheorem{ntn}[lem]{Notation}
\newtheorem{cor}[lem]{Corollary}
\newtheorem{obs}[lem]{Observation}
\newtheorem{prop}[thm]{Proposition}

\definecolor{bubbles}{rgb}{0.91, 1.0, 1.0}
\definecolor{aqua}{rgb}{0.0, 1.0, 1.0}
\definecolor{blue-green}{rgb}{0.0, 0.87, 0.87}
\definecolor{amethyst}{rgb}{0.6, 0.4, 0.8}
\definecolor{lightgoldenrodyellow}{rgb}{0.98, 0.98, 0.82}

\DeclareRobustCommand{\hldana}[1]{{\sethlcolor{bubbles}\hl{#1}}}
 \DeclareRobustCommand{\hltalya}[1]{{\sethlcolor{lightgoldenrodyellow}\hl{#1}}}
\DeclareRobustCommand{\hlwill}[1]{{\sethlcolor{pink}\hl{#1}}}
\ifnum\withcomments=1
\newcommand{\dnote}[1]{\hldana{[\textbf{D}: #1]}}

\newcommand{\tnote}[1]{\hltalya{[\textbf{T}: #1]}}
\newcommand{\wnote}[1]{\hlwill{[\textbf{W}: #1]}}

\newcommand{\old}[1]{\textcolor{black!70}{#1}}
\else
\newcommand{\dnote}[1]{}

\newcommand{\tnote}[1]{}
\newcommand{\wnote}[1]{}

\newcommand{\old}[1]{}
\fi

\renewcommand{\epsilon}{\eps}

\newcommand{\eps}{\varepsilon}

\newcommand{\mQ}{\mathcal{Q}}

\newcommand{\EX}{{\mathrm{Ex}}}

\renewcommand{\th}{^{\textrm{th}}}
\newcommand{\tP}{P}
\newcommand{\cP}{\widecheck{P}}

\newcommand*{\medcup}{\textstyle \bigcup}
\newcommand*{\medcap}{\textstyle \bigcap}

\newcommand{\reduce}{\hyperref[reduce]{\textup{\color{black}{\sf Peel-With-Reduced-Error}}}}
\newcommand{\peel}{\hyperref[peel]{\textup{\color{black}{{\textsf{Peel}}}}}}
\newcommand{\est}{\hyperref[est]{\textup{\color{black}{\sf Estimate-Arboricity}}}}
\newcommand{\peelv}{\hyperref[peelv]{\textup{\color{black}{\sf Peel-Vertex}}}}

	\newcommand{\poly}{poly}

\renewcommand{\deg}{{\sf deg}}

\newcommand{\mA}{\mathcal{A}}

\newcommand{\alphaX}{100 \log^2 n \cdot \alpha}

\newcommand{\YES}{{\mathsf{Yes}}}
\newcommand{\NO}{{\mathsf{No}}}

\newcommand{\calE}{\mathcal{E}}

\newcommand{\arb}{{\mathsf{arb}}}
\newcommand{\degen}{{\mathsf{degen}}}
\newcommand{\dens}{{\mathsf{dens}}}
\newcommand{\ta}{\tilde{\alpha}}
\newcommand{\mE}{{\mathcal{E}_s}}

\newcommand{\Nei}[1]{A_{#1}}
\newcommand{\tNei}[1]{B_{#1}}
\newcommand{\cNei}[1]{\widecheck{A}_{#1}}

\newcommand{\tH}{{H}}
\newcommand{\tp}{{q}}

\newcommand{\ap}{\widecheck{q}}

\newcommand\numberthis{\addtocounter{equation}{1}\tag{\theequation}}

\newcommand{\algo}[3]{
\begin{figure}[ht!]
	\centering
	\fbox{
		\begin{minipage}{0.95\textwidth}
		{\fontfamily{cmss}\selectfont
			{\sansmath
							#1}
}			
			\end{minipage}
		}
	\caption{#2}\label{#3}
	\end{figure}
}	

\newcommand{\comment}[1]{\hspace*{\fill} \textcolor{blue}{$\rhd$ #1}}

\begin{document}


\begin{titlepage}

\title{Approximating the Arboricity in Sublinear Time\footnote{A conference version of this manuscript is to appear in SODA 2022.}}
\captionsetup[subfigure]{labelformat = parens, labelsep = space, font = small}
	\author{Talya Eden  \thanks{CSAIL at MIT,  Boston University Department of Computer Science, \textit{talyaa01@gmail.com}. Partially supported by the NSF Grant CCF-1740751, the Eric and
		Wendy Schmidt Fund,  Ben-Gurion University, and the Computer Science Department at Boston University.} \\ Boston University and MIT
	\and Saleet Mossel \thanks{CSAIL at MIT, \textit{saleet@mit.edu}.}\\ MIT 
	\and 	Dana Ron \thanks{Tel Aviv University, \textit{danaron@tau.ac.il}. Partially supported  by the Israel Science Foundation (grant No.~1041/18).} \\ Tel Aviv University
}
\date{\vspace{-5ex}}
\maketitle

\begin{abstract}
We consider the problem of approximating the arboricity of a graph $G= (V,E)$, which we denote by $\arb(G)$, in sublinear time, where the arboricity of a graph is the minimal number of forests required to cover its edges.
An algorithm for this problem may perform degree and neighbor queries, and is allowed a small error probability.
We design an algorithm that outputs an estimate $\hat{\alpha}$, such that with probability $1-1/\poly(n)$, $\arb(G)/c\log^2 n \leq \hat{\alpha} \leq  \arb(G)$, where $n=|V|$ and $c$ is a constant.
The expected query complexity and running time of the algorithm are
 $O(n/\arb(G))\cdot \poly (\log n)$, and this upper bound also holds with high probability. 
 This bound is optimal for such an approximation  up to a $\poly(\log n)$ factor.

\end{abstract}

\thispagestyle{empty}
\end{titlepage}

\setcounter{page}{1}
\newpage

		\tikzset{dotted pattern/.style args={#1}{
		postaction=decorate,
		decoration={
			markings,
			mark=between positions 0.25 and 0.75 step 0.25 with {
				\fill[radius=#1] (0,0) circle;
			}
		}
	},
	dotted pattern/.default={1pt},
}

\pgfmathsetmacro{\minSize}{.7cm}

\pgfdeclarelayer{nodelayer}
\pgfdeclarelayer{edgelayer}
\pgfsetlayers{edgelayer,nodelayer}

\colorlet{lightGray}{gray!30}

\tikzstyle{ghost}=[]
\tikzstyle{reg}=[draw,circle,fill=white,minimum size=\minSize]
\tikzstyle{peeled}=[draw,circle,fill=gray!80,minimum size=\minSize, pattern=crosshatch dots]
\tikzstyle{pruned}=[draw,circle,fill=blue!20,minimum size=\minSize]
\tikzstyle{inactive}=[draw=lightGray,circle,fill=lightGray!50,minimum size=\minSize,dotted]

\newcommand{\declareU}[3]{

	\begin{pgfonlayer}{nodelayer}
		\node [style=reg] (0) at (0,0) {$v$};
	\end{pgfonlayer}
	\foreach \i in {1,...,7} { 
		\begin{pgfonlayer}{nodelayer}
			\ifnum\i<5
				\node [#3] (u\i) at (-4+\i, -1.5) {};
			\else
				\node [#1] (u\i) at (-4+\i, -1.5) {};
			\fi
		\end{pgfonlayer}
		\begin{pgfonlayer}{edgelayer}
			\draw[#2] (u\i.center) to (0.center);
		\end{pgfonlayer}
		}
}

\newcommand{\drawW}[3]{

	\foreach \i in {1,...,12} { 
	\begin{pgfonlayer}{nodelayer}
		\node [#1] (w\i) at (-4+0.7*\i, -3) {};
	\end{pgfonlayer}
}

\begin{pgfonlayer}{edgelayer}
	
	\draw[#2] (u3.center) to (w1.center);
	\draw[#2] (u3.center) to (w2.center);
	\draw[#2] (u3.center) to (w3.center);
	
	\draw[#2] (u4.center) to (w4.center);
	\draw[#2] (u4.center) to (w5.center);
	\draw[#2] (u4.center) to (w6.center);
	
	\draw[#3] (u5.center) to (w7.center);
	
	\draw[#3] (u6.center) to (w8.center);	
	\draw[#3] (u6.center) to (w9.center);
	\draw[#3] (u6.center) to (w10.center);
	
	\draw[#3] (u7.center) to (w11.center);
	\draw[#3] (u7.center) to (w12.center);
	
\end{pgfonlayer}

}

\newcommand{\recolorU}{
	\begin{pgfonlayer}{nodelayer} 
	\node [peeled] (u1) at (u1) {};
	\node [peeled] (u2) at (u2) {};	
	\node [pruned] (u3) at (u3) {};
	\node [pruned] (u4) at (u4) {};
\end{pgfonlayer}	

}

\newcommand{\justVNew}{
\begin{tikzpicture}
    \declareU{ghost}{dashed}{ghost}
	\node[ghost]  (w1) at (0, -3) {};
\end{tikzpicture}
}

\newcommand{\treeDepthOne}{
	\begin{tikzpicture}

	\declareU{reg}{black}{reg}
	\recolorU
	\drawW{ghost}{dashed}{dashed}
		
	\end{tikzpicture}
}
	
\newcommand{\treeDepthTwo}{
\def \inactive {1}

	\begin{tikzpicture}
\begin{scope}[xscale=1.5]

\begin{scope}[xscale=.8]
	\declareU{reg}{black}{inactive}
\end{scope}
	\drawW{ghost}{dashed, lightGray}{black} 
	
	\foreach \i in {7,...,12} { 
	\begin{pgfonlayer}{nodelayer}
		\node [reg] (w\i) at (w\i) {};
	\end{pgfonlayer}
	}
	\begin{pgfonlayer}{nodelayer}
		\node [reg] (w6) at (w6) {};
	\end{pgfonlayer}
	
	\begin{pgfonlayer}{edgelayer}
		\coordinate (d) at ($ (w6) + (-1,0) $);
		\draw [white,thick] (w6.center) to (u4.center) ;
		\draw [dashed, lightGray] (u4.center) to (d) ; 		\draw [black] (u5.center) to (w6.center) ;
	\end{pgfonlayer}	

	\foreach \i in {1,...,14} { 
		\begin{pgfonlayer}{nodelayer}
			\node [] (z\i) at (-.5+0.4*\i, -4.5) {};
		\end{pgfonlayer}
	}

\begin{pgfonlayer}{nodelayer} 
	\node [peeled, pattern=north east lines] (u5) at (u5) {};
	\node [pruned,preaction={fill, blue!50}, pattern=  north west lines] (u6) at (u6) {};

	\node [peeled] (w6) at (w6) {};
	\node [pruned] (w7) at (w7) {};
	\node [peeled] (w8) at (w8) {};
	\node [pruned] (w9) at (w9) {};
	\node [pruned] (w11) at (w11) {};

	\foreach \i in {1,...,4} {
		\begin{pgfonlayer}{edgelayer}
			\draw[white,thick] (u\i.center) to (0.center);
			\draw[lightGray,dashed] (u\i.center) to (0.center);
		\end{pgfonlayer}
	}

\end{pgfonlayer}

\begin{pgfonlayer}{edgelayer}

	\draw[dashed] (w7.center) to (z1.center);
	\draw[dashed] (w7.center) to (z2.center);
	\draw[dashed] (w7.center) to (z3.center);
	\draw[dashed] (w7.center) to (z4.center);
	
	\draw[dashed] (w9.center) to (z5.center);
	\draw[dashed] (w9.center) to (z6.center);
	
	\draw[dashed] (w10.center) to (z7.center);
	\draw[dashed] (w10.center) to (z8.center);
	\draw[dashed] (w10.center) to (z9.center);
	\draw[dashed] (w10.center) to (z10.center);

	\draw[dashed] (w11.center) to (z11.center);
	\draw[dashed] (w11.center) to (z12.center);

	\draw[dashed] (w12.center) to (z13.center);
	\draw[dashed] (w12.center) to (z14.center);

\end{pgfonlayer}

\end{scope}
	
\end{tikzpicture}

}

\newcommand{\drawHere}{
\begin{figure}[ht!]
	\centering
	\caption{The invocation of the (modified) approximate peeling procedure on vertex $v$ with indices $j=0$ (top left), $j=1$ (top right), and $j=2$ (bottom).} \label{fig:illustratePeel}
	\begin{subfigure}[t]{0.45\textwidth}
		\centering
		\vspace{.5cm}
		\justVNew
		\renewcommand{\thefigure}{1a}
		\captionof{figure}{For $j=0$, a single degree query on $v$ is performed. The returned degree, $d(v)$, determines the $1$-cost of $v$, which is $d(v)/\alpha$.
\label{j=0}}
	\end{subfigure}\hfill
	\begin{subfigure}[t]{0.53\textwidth}
		\centering			\vspace{.5cm}
		\treeDepthOne			
		\captionof{figure}{For $j=1$, $d(v)/\alpha$ neighbors of $v$ are sampled,
and for each sampled neighbor $u$, the procedure is recursively  invoked on $u$ with index $j=0$. Once all the invocations on the sampled neighbors are completed, $v$ discards its peeled vertices (in dotted gray), and prunes the costly ones (in purple), namely those with highest $1$-cost. \label{j=1}}
	\end{subfigure}\hfill
	\begin{subfigure}{0.9\textwidth}
		\centering \vspace{0.2cm}
		\treeDepthTwo
		\captionof{figure}{For $j=2$, $v$ recursively calls the procedure on each of its remaining active vertices (those in white in Subfigure~\ref{j=1}) with $j=1$. In turn, each such $u$ samples 
$d(u)/\alpha$ neighbors
and recursively invokes the procedure on its set of sampled neighbors with $j=0$. Once these recursive calls return, $v$ discards its (newly) peeled neighbors (in striped gray) and prunes  the ($2$-)costly neighbors (in striped blue).
\label{j=2}}
	\end{subfigure}

\end{figure}
}

	\section{Introduction}

The  arboricity of a graph $G$, denoted $\arb(G)$, is a measure of its density ``everywhere''.
Formally, it is defined as the minimum number of forests into which its edges can be partitioned, and it holds~\cite{NW1,Tutte,NW2}  that
 $\arb(G)=\max_{S\subseteq V}\left\{\left\lceil \frac{|E(S)|}{|S|-1}\right\rceil\right\}$, where $E(S)$ denotes the set of edges in the subgraph of $G$ induced by $S$.\footnote{It is also closely related to the degeneracy, $\degen(G)$,  and maximum subgraph density, $\dens(G)$,  of the graph. The degeneracy of a graph $G$ is the smallest integer $k$ such that in every subgraph of $G$ there is vertex of degree at most $k$, and the maximum subgraph density is $\max_{S\subseteq G}\{|E(S)|/|S|\}$. It holds that
$\arb(G)\leq \degen(G)\leq 2\arb(G)-1$, and $\dens(G)\leq \arb(G)\leq \dens(G)+1$.}

Arboricity is not only a basic measure, but also plays an important role in designing efficient algorithms, including, but not limited to: listing subgraphs, e.g., ~\cite{CN85,cai2006random, BeraPS20, gishboliner2020counting, lior_counting_cycles, bressan2021faster}, graph coloring, e.g., ~\cite{BE,kothapalli2011distributed, parter2016local,solomon2019improved, ghaffari2019distributed, henzinger2020explicit}, and
maintaining small representations, e.g.,~\cite{brodal1999dynamic,kowalik2007adjacency,he2014orienting}.
Furthermore, several NP-hard problems such as {\sc Clique}, {\sc Independent-Set} and {\sc Dominating-Set}  become fixed-parameter tractable in bounded arboricity graphs~\cite{AG09, GV08, ELS13, LPW13, BU17}).

 In the sublinear-time regime, when $\arb(G)$ is bounded, there exist improved algorithms for approximating the number of cliques~\cite{eden2020faster},  approximating the moments of the degree distribution~\cite{eden2019sublinear}, and sampling edges and cliques almost uniformly at random~\cite{eden2019arboricity, eden2020almost}. All these algorithms require receiving an upper bound on the arboricity as input in order to achieve the improved results.

The arboricity of a graph can be exactly computed in polynomial time~\cite{edmonds1965minimum,picard1982network}, where the fastest algorithm is due to Gabow~\cite{gabow1998algorithms}  and runs in time $O(m^{3/2}\log(n^2/m))$, where $n$ and $m$ denote the number of vertices and edges, respectively, in the graph. For a comprehensive list of results on exactly computing the arboricity  see~\cite{blumenstock2019constructive}.
Several $O(n+m)$ time algorithms exist for computing a $2$-factor approximation of the arboricity~\cite{eppstein1994arboricity,arikati1997efficient,charikar2000greedy}.\footnote{A $k$-factor approximation of the arboricity is a value $\tilde{\alpha}$ so that $\arb(G)/k\leq \tilde{\alpha}\leq \arb(G)$.}

A natural question is whether the arboricity can be approximated \emph{much more efficiently}, and in particular, in sublinear time. Specifically, we consider the incidence-list query model, which allows for degree and neighbor  queries.\footnote{A degree query on a vertex $v$ returns the degree of $v$, $\deg(v)$, and a neighbor query on $v$ with an index $i \leq \deg(v)$ returns the $i\th$ neighbor of $v$.}
For the closely related problem of finding the densest subgraph,  Bhattacharya et al.~\cite{BHNT15}  showed that  their  $2$-factor approximation $\widetilde{O}(n)$-space\footnote{Throughout the paper, we use  $\widetilde{O}(\cdot)$ and $\widetilde{\Omega}(\cdot)$ to suppress  $\poly(\log n)$ factors.} dynamic streaming algorithm can be  adapted to run in $\widetilde{O}(n)$ time in the incidence-list  model. In a follow up work, McGregor et al.~\cite{mcgregor2015densest} improved the approximation factor to  $(1+\eps)$, and it can be shown that their algorithm can also be adapted to run in $\widetilde{O}(n)$ time in the incidence-list model.
In~\cite{bahmani2012densest}, Bahmani et al. proved a lower bound for the streaming variant of the problem.
In \cite{BHNT15}, the authors adapted this lower bound to the incidence-list model, and showed that for graphs with  arboricity $\Theta(k)$,
any $\Theta(k)$-factor approximation algorithm must perform
$\Omega(n/k^2)$ queries.

In this work we ask:

\vspace{.2cm}
\emph{Is it possible to go below time linear in $n$, when the arboricity is super-constant?}
\vspace{.2cm}

We present an algorithm, \est$(G)$,  that computes an $O(\log^2 n)$-factor approximation  of the arboricity of $G$ in sublinear time.

\begin{restatable}{thm}{upperBound}\label{thm:upperBound}
	There exists an algorithm, \est$(G)$, that  with probability at least $1-O(1/n^2)$ returns a value $\widehat{\alpha}$, such that   \[ \arb(G)/(200 \log ^2 n) \leq \widehat{\alpha} \leq \arb(G) .\]
	The expected query complexity and running time  of the algorithm are $O(n\log^3 n/\arb(G))$, and this also holds with probability $1-O(1/n^2)$.
\end{restatable}

A different setting of the parameters in the lower bound construction that appears in the full version~\cite{BHNT15} of~\cite{BHNT15}, gives the following.

\begin{prop}[Adaptation of Theorem 7.3 in~\cite{BHNT15}]\label{prop:lb}
Any algorithm that, with probability at least $2/3$
returns a
$k$-factor approximation of the arboricity of a graph $G$, must perform $\Omega(n/(k\cdot \arb(G)))$ queries.	
\end{prop}
Hence, our algorithm's query complexity is optimal, up to $\poly(\log n)$ factors.	
	
Compared to the $\tilde{O}(n)$-time, $(1+\eps)$-approximation algorithm of~\cite{mcgregor},
our algorithm improves the time complexity by a factor of  $O(\arb(G))$, at the cost of increasing the approximation factor to  $O(\log^2(n))$.			
A natural question
is whether the $O(\log^2 n)$ approximation factor can be improved  with similar time complexity.

We note that the related problem of \emph{tolerant testing} of bounded arboricity (in the incidence-list model) was studied by Eden, Levi and Ron~\cite{eden2020testing}: They proved that graphs that are $\gamma$-close\footnote{In the incidence-list model, a graph $G$ is said to be $\gamma$-close to some property $\Pi$, if there exists a graph $G'\in \Pi$ such that
	we can get $G'$ from $G$
	by at most $\gamma |E(G)|$ edge deletions and insertions to $G$.} to having arboricity at most $\alpha$, can be distinguished from graphs that are $20\gamma$-far from having arboricity at most $3\alpha$.
The running of their algorithm is $\widetilde{O}\left(\frac{n}{\gamma\sqrt m}+ (1/\gamma)^{O(\log(1/\gamma))}\right)$.
They further showed that their algorithm can be used to estimate what they refer to as the \emph{corrected arboricity} of $G$,  $ \alpha^*(G)=\min_{G'}\{ \arb(G') \mid \text{$G'$ is $\gamma$-close to $G$}\}$. Observe that this value might be much smaller than $\arb(G)$.\footnote{To see that $\alpha^*(G)$ might be much smaller than $\arb(G)$, consider a graph  $G$  consisting of a clique of size $\sqrt{\gamma m}$ and of a set of $n-\sqrt{\gamma m}$ vertices, each of degree $(1-\gamma)m/n$. Then $G$ has arboricity $\Theta(\sqrt{\gamma m})$, but the corrected arboricity is $\Theta(m/n)$.} In particular, the corrected value of the arboricity cannot be used in the aforementioned sublinear algorithms~\cite{eden2019sublinear,eden2019arboricity,eden2020faster} that rely on receiving an upper bound on the arboricity (whereas the  estimate output by our algorithm can).

Finally, we observe that our  algorithm can be adapted to the streaming model, providing lower space complexity compared to the state of the art~\cite{mcgregor2015densest}, at the cost of increasing the approximation factor as well as  the number of passes. For further discussion on the relation between the results, see Section~\ref{sec:stream}.

\begin{restatable}{thm}{ubStreaming}\label{thm:upperBoundStreaming}
	Given a lower bound $\alpha<\arb(G)$, algorithm \est\ can be implemented in the streaming model, using $O(\log n)$-passes and $\widetilde{O}(n/\alpha)$ space in expectation. With probability at least $1-O(1/n^2)$, the algorithm   outputs a $200\log^2n$-factor approximation of $\arb(G)$.
\end{restatable}

\subsection{A high-level description of the algorithm}
In what follows we consider the task of distinguishing between the case that $\arb(G)\leq \alpha$ and the case that $\arb(G) > \rho \alpha$, for a given arboricity parameter $\alpha$, and a bounded approximation ratio $\rho$.
That is,  if $\arb(G)\leq \alpha$, then it should output $\YES$,
and if $\arb(G) > \rho \alpha$, then it should output $\NO$, and it is allowed a small error probability.
(If $\alpha < \arb(G) \leq \rho(\alpha)$, then the algorithm may output either $\YES$ or $\NO$.)
Once we design an algorithm for this  promise problem, we can search for $\alpha$ using standard techniques.

\paragraph{Vertex layering.}\label{sec:vertex_layering}
Our starting point is the fact
that if a graph $G=(V,E)$ has arboricity at most $\alpha$, then $V$ can be partitioned into $\ell = O(\log n)$
\emph{layers}, $L_0,\dots,L_\ell$, where each vertex in layer $L_i$ has at most $3\alpha$ neighbors\footnote{One can replace the constant $3$ by any constant bigger than $2$ - for simplicity, we present the layering with the constant $3$.} in layers $L_j$, $j\geq i$.
 This partition is due to Barenboim and Elkin~\cite{BE}, and was part of their algorithm for computing a forest decomposition of a graph in the distributed setting, which itself is used to obtain  efficient coloring and Maximal Independent Set algorithms.

 Such a partition can be obtained by what we refer to as a \emph{peeling} process. First, all vertices with degree at most $3\alpha$ are put in $L_0$, and are removed (peeled) from the graph. Then, the updated degree of all vertices is computed, and all the vertices with updated degree at most $3\alpha$ are put in $L_1$, and peeled from $G$. The process continues, and it can be shown that in graphs with arboricity at most $\alpha$, in every iteration, at least a constant fraction of the vertices is peeled. Hence, the process terminates after $\ell=O(\log n)$ iterations (and each vertex belongs to some $L_i$, $i \leq \ell$).

On the other hand, if $\arb(G)>\rho\alpha$ for $\rho\geq 3$, then there exists a subset $R\subseteq V$, such that every vertex in $R$ has more than $\rho\alpha$ neighbors in $R$.
 This implies that no vertex in $R$ will ever be peeled (and added to a layer $L_i$).
Hence, the vertices of the set $R$ are ``witnesses'' to the fact that $G$ has arboricity greater than $\rho\alpha$, and in order to determine  that $\arb(G)>\rho\alpha$, we will be interested in detecting at least one vertex from the set $R$.
We shall set $\rho$ subsequently, but for now we assume that it is sufficiently larger than 3.

Consider taking a sample $X$ of $O(n/\alpha)$ vertices (uniformly at random). If $\arb(G) > \rho\alpha$, then
we expect the sample to contains at least one vertex from the aforementioned subset $R$. On the other hand, if $\arb(G) \leq \alpha$, then every sampled vertex belongs to some $L_i$, $i\leq \ell$. In order to distinguish between the two cases, we would like to run an \emph{approximate peeling procedure} on each sampled vertex. The intention is that if $\arb(G)\leq \alpha$, then this  procedure determines for each vertex $v$ that it belongs to some $L_i$, $i\leq \ell$,
while if $\arb(G)>\rho\alpha$ then for $v\in R$, the procedure determines that $v$ does not belong to any $L_i$, $i\leq \ell$.

\paragraph{An iterative $+$ recursive peeling process.}
For each vertex $v$ in the sample $X$, we would like to decide whether $v$ belongs to some $L_i$, $i\leq \ell$ or not.
To this end we run in at most $\ell+1$ iterations, indexed by $j$, starting from $j=0$. In iteration $j$, we aim  to determine which vertices in the sample belong to $L_j$. This is done by calling a (recursive approximate) peeling procedure on each $v$ in the sample that was not peeled in previous iterations, with the parameter $j$.
For $j=0$, this is an easy task, as it only requires performing a degree query on $v$ and peeling $v$ (placing it in $L_0$) if $d(v) \leq 3\alpha$.

For $j=1$, the procedure samples $d(v)/\alpha$ neighbors of $v$, and recursively invokes  itself with $j=0$ on each of the sampled neighbors. This results in a partial BFS tree of depth 1 rooted at $v$, where for each sampled neighbor $u$ of $v$, we know whether it belongs to $L_0$ (and was hence peeled). If the number of un-peeled sampled neighbors (children) of $v$ is below a certain threshold $\tau$ (that will depend on $\alpha,\rho$ and $j$), then $v$ is peeled (and deemed to belong to $L_1$).
Observe that if $v$ indeed belongs to $L_1$, then, since it has at most $3\alpha$ neighbors that do not belong to $L_0$, we expect to see only a constant number of such vertices among $v$'s sampled neighbors.
In such a case, the approximate peeling procedure peels $v$ (so that  $v$ is detected as belonging to $L_1$).
On the other hand, if $v$ has a significantly larger number of neighbors that do not belong to $L_0$, then $v$ is not peeled.

In general, for $j> 1$, if $v$ was not yet peeled in previous iterations,  a partial BFS tree of depth $j-1$ was already constructed for $v$.\footnote{If the same vertex is encountered more than once, then in terms of the tree structure, we maintain two (or more) copies of the vertex.}
Considering the children of $v$ in the tree, some were peeled (and deemed to belong to $L_0,L_1,\dots,L_{j-2}$), and some  are (yet) un-peeled.
The procedure is now invoked recursively on each of these un-peeled children of $v$ with the parameter $j-1$, to decide for each of them whether it is deemed to belong to $L_{j-1}$ and hence should be peeled.
If, after these recursive calls, the number of yet un-peeled neighbors of $v$ is sufficiently small (below the threshold $\tau$), then $v$ is peeled (and deemed to belong to $L_j$).

By the above description, if $\arb(G)\leq \alpha$, then we expect that for every $i \leq \ell$ and every $v\in L_i$, $v$ will be peeled after $j \leq i$ iterations. On the other hand, if $\arb(G)>\alphaX$, we expect that at least one vertex in
the sample  $X$ will not be peeled after all $\ell$ iterations.

\paragraph{Error probability and query complexity.}
One issue that needs to be addressed in the above description, is bounding the error probability (due to sampling). This can be handled by standard probabilistic analysis (where we set the peeling threshold $\tau$ to $O(\log n)$, which implies that the approximation factor $\rho$ must be larger).
This leaves us with the central issue of the query complexity (and running time) of the algorithm.
We would like to show that when $\arb(G) \leq \alpha$, we can bound,  with sufficiently high probability, the total number of queries performed until all sampled vertices are peeled, by $\widetilde{O}(n/\alpha)$. This will allow us to terminate the algorithm if the number of queries exceeds this upper bound (as we have an indication that $\arb(G) > \alpha$).

To this end we shall actually modify the approximate peeling procedure, but before describing this modification, we provide some more intuition. Our focus for now is on the case that $\arb(G) \leq \alpha$. We later show that this modification (for the sake of upper bounding the complexity) does not have a significant effect on the error probability. That is, it still holds that when $\arb(G) \leq \alpha$,  every sampled vertex is peeled (with high probability) in some iteration $j\leq \ell$, while
 when $\arb(G) > \rho\alpha$, vertices in $R$ will not be peeled (with high probability).

\paragraph{A special case (and some wishful thinking).}
Consider the following graph. The graph vertices are partitioned into $t=O(\log n)$ subsets, $V_0,\dots, V_t$.
For each $i$, $|V_i|$ is roughly $n/\rho^{i}$. All vertices in $V_0$ have degree $\alpha$, and all other vertices have degree  $(\rho+1) \alpha$.
The edges in the graph are all between consecutive subsets, $V_i$, $V_{i+1}$, where, each vertex in $V_{i}$ has $\alpha$ neighbors in $V_{i+1}$, and the remaining neighbors in $V_{i-1}$ (for $i>0$). By the definition of the peeling process, $L_i=V_i$.

When taking a uniform sample of $O(n/\alpha)$ vertices, we expect to get $O(n/(\rho^i \alpha))$ vertices from each $V_i$.
Suppose that, when sampling the $d(u)/\alpha=\rho$ neighbors of any vertex $u\notin L_0=V_0$ as part of the approximate peeling process, we always get neighbors that belong to the layer below. Then the number of queries performed (including in recursive calls to the procedure) until a vertex $v\in V_i$ is peeled, is $O(\rho^i)$.
We hence get a total of
$O\left(\frac{n}{\alpha}\cdot \sum_i \frac{1}{\rho^i}\cdot \rho^i\right)= \widetilde{O}\left(n/\alpha\right)$ queries.
Unfortunately, we cannot assume that our input graph has such a convenient layered structure. Furthermore, we 
cannot rely on the (wishful-thinking) assumption that the samples of neighbors  contain only neighbors that belong to lower layers. 

One main building block of our analysis is showing that for any graph $G$ such that $\arb(G)\leq \alpha$, if neighbors are sampled only from lower layers, then we can still get an upper bound of $\widetilde{O}(n/\alpha)$  as in  the special case of the graph described above. Here we shall not elaborate on the proof of this claim, but rather focus on how to modify the approximate peeling procedure so as to obtain a similar upper bound, without relying on this (wishful-thinking) assumption.

\drawHere

\paragraph{Modifying the approximate peeling procedure.}
As discussed above,
when sampling neighbors of a vertex $v$ in the course of an invocation of the approximate peeling procedure,
we would have liked
to be able to identify those sampled neighbors of $v$ that belong to higher layers, so as to avoid performing recursive calls on them. 
The reason is, that such calls may be too costly in terms of the query complexity.
But this implies that it is not really necessary to exactly identify all such higher-layer neighbors of $v$, but rather only those for which the recursive invocation of the approximate peeling procedure will have a high cost in terms of the query complexity. We shall hence be interested in identifying such costly vertices and  avoid performing recursive calls on them.  We refer to such a process as ``pruning'' and to vertices that are neither peeled nor pruned as \emph{active}.

Fortunately, as we discuss shortly,
when we call the procedure with a  neighbor $u$  of $v$  and  an index $j-1$, assuming $u$ is not peeled as a result of this call, we can already \emph{exactly} compute the ``future" cost of invoking the recursive procedure on $u$ with index $j$. We refer to this value as the \emph{$j$-cost} of $u$.
This allows us to prune those $\tau$ active neighbors of $v$ that have highest $j$-costs, so that the procedure will not be called recursively on them if it is invoked with $v$ and $j+1$.  For an illustration of the modified process, see Figure~\ref{fig:illustratePeel}.

\paragraph{Computing the $j$-cost of a vertex.}
First observe that the $0$-cost of any vertex is always $1$, since for $j=0$, only a single degree query is performed.
We next explain how the $j$-cost of a vertex can be computed for $j>0$, when the procedure is invoked on this vertex with parameter $j-1$.
The complexity of invoking the procedure on a vertex  $u$  with $j=1$ is determined by the number of its sampled neighbors, $d(u)/\alpha$. Hence, this number can already be computed once $d(u)$ is determined (when the procedure is invoked on $u$ with $j=0$), without actually identifying the neighbors themselves (that is, without performing any neighbor queries).

	For general $j$, the $j$-cost of $u$ is computed as follows  (when the procedure is invoked on $u$  with the index $j-1$). First, the  procedure is invoked recursively on the remaining active neighbors of $u$. Once these  calls return, the peeled and pruned (according to their $(j-1)$ costs) neighbors are removed from the set of active neighbors of $u$, and the $j$-cost of $u$ is set to be the sum of $(j-1)$-costs over the updated set of active neighbors.

\paragraph{Wrapping things up.}
 Consider first the case that $\arb(G)\leq \alpha$. Then with high probability, for every vertex on which the procedure is invoked, the number of its sampled neighbors from higher levels is not much larger than $\tau$. Conditioned on this event, the following holds. The  total number of queries performed by the modified procedure on the initially sampled $O(n/\alpha)$ vertices, can be  upper bounded by the number of queries that would have been performed by the original procedure when conditioning on sampled neighbors only belonging to lower levels.
 The algorithm outputs $\YES$ if all sampled vertices are peeled (by the modified approximate peeling procedure) after at most $\ell$ levels of recursion, and the total number of queries performed (when sampling random neighbors), is not much larger than the aforementioned upper bound.
 Otherwise, it outputs $\NO$.

The correctness of the algorithm follows for the case that $\arb(G) \leq \alpha$, since
 the pruning performed by the modified procedure can only increase the probability that a vertex $v\in L_i$ will be peeled in at most $i$ levels of recursion.
  Turning to the $\NO$ instances, we  set  the approximation factor $\rho$ to $O(\log^2 n)$. Recall that the peeling and pruning threshold, $\tau$, is $O(\log n)$, and the maximum number of allowed recursion levels is  $\ell = O(\log n)$. Therefore, for such a setting of $\rho$, if $\arb(G) > \rho \alpha$, and the algorithm samples some vertex $v\in R$, then (with high probability) $v$ will not be peeled in $\ell$ levels of recursion, so that the algorithm will output $\NO$, as required.

\subsection{Comparing our algorithm to previous work}\label{sec:compare}

\subsubsection{Related algorithms in other models of computation.}
\sloppy
 Elkin and Barenboim~\cite{BE} design a distributed algorithm that computes the layering described in Section~\ref{sec:vertex_layering} with a threshold of $2(1+\eps)$.
 The round complexity of their algorithm  is $O(\log_{1+\eps} n)$.
Bahmani, Kumar and Vassilvitskii~\cite{bahmani2012densest}
implement the same peeling algorithm in the streaming model,
and output a subgraph that preserves the maximum density in the graph, up to a factor of $2(1+\eps)$.
Their algorithm performs
$O(\log_{1+\eps} n)$-passes and uses $O(n)$-space.

The algorithm by Bhattacharya et al.~\cite{BHNT15}, has the same output guarantees as the  algorithm of~\cite{bahmani2012densest},  while performing only a single pass over the stream. This comes at the cost of increasing the space complexity to $O(n\cdot \poly(\log_{1+\eps}n))$. The algorithm uses the crucial observation, that the densest subgraph remains densest even if each edge is sub-sampled with probability $d$. Hence, by taking a sample of $O(m/d)$ edges, one can implement the peeling procedure on the sampled subgraph, and get a $2(1+\eps)$-approximation.

 Our algorithm also has an element of edge sampling in the form of neighbor sampling. But as opposed to~\cite{BHNT15} in which the neighbor sampling is performed on \emph{all} vertices, in our algorithm  it  is performed only on an initial  set of $O(n/\arb(G))$ vertices, and their sampled descendants.  Our main challenge is in showing how the peeling procedure can be modified in order to obtain an upper bound of   $\widetilde{O}(n/\arb(G))$ on the total number of queries performed.

The algorithm by McGregor et al.~\cite{mcgregor2015densest} relies on the same sub-sampling as in~\cite{BHNT15}, but instead of computing the layering, it directly computes the densest subgraph in $G'$ and  prove that it preserves the density up to a factor of $(1+\eps)$. We note that in both these results the main focus was on achieving fast update time per each edge insertion/deletion, which is of less relevance to us.

\subsubsection{Related algorithms in the incidence-list model}

As noted previously, Bhattacharya et al.~\cite{BHNT15}
also describe  an adaption of their streaming algorithm  to the incidence-list model, where
in order to implement the iid edge samples,  $O(n)$ degree queries are performed, resulting in complexity $\widetilde{O}(n)$.

 Eden, Levi and Ron~\cite{eden2020testing} present a tolerant testing algorithm for arboricity.
In the course of their algorithm, they too perform a certain kind of approximate peeling process based on neighbor sampling. However, as their end result is quite different from ours, their algorithm, and its analysis, differ as well. In particular, recall that in the context of testing, a graph is considered $\eps$-close to having arboricity at most $\alpha$, if it can be made to have arboricity $\alpha$ by removing at most an $\eps$-fraction of its edges. This implies that for each vertex, it suffices to take a sample of size $O(1/\eps)$ of its neighbors, independently of its degree.
 This in turn implies that the constructed ``approximate peeling trees'' have degree and depth depend only on $1/\eps$ (indeed, this is the source of the term $(1/\eps)^{O(\log(1/\epsilon))}$ in the complexity of their algorithm).

\section{Preliminaries}\label{sec:prel}

Let two integers $i\leq j$, let $[i,j]$ denote the set of integers $i\leq k\leq j$. If $i=1$, then we use the shorthand $[j]$ for $[0,j]$.
	
We consider simple undirected graphs $G=(V,E)$ where $|V|=n$ and $|E|=m$. Let $\Gamma(v)$  denote the set of neighbors of a vertex $v$, and $d(v)=|\Gamma(v)|$.	For a subset of vertices $S$, we use $d_S(v)$ to denote the degree of $v$ in the subgraph induced by $S$.
We abuse notation and use set operations to manipulate multisets, where cardinality of the set is the sum of multiplicities of its elements, and other operations are the natural generalizations of the set operations. Where the distinction between sets and multisets is not important, we might simply refer to multisets as sets.

		\begin{dfn}
			The {\sf arboricity} of a graph $G=(V,E)$, denoted $\arb(G)$, is the minimum number of forests into which $E$ can be partitioned.
		\end{dfn}

	\begin{thm}[Multiplicative Chernoff Bound]\label{thm:chernoff}	
		Let $\chi_1, \ldots, \chi_k$ be independent random variables in $\{0,1\}$. Let $\chi=\sum_{i=1}^k \chi_i,$ and $\mu=\EX[\chi]$. Then
		$$\Pr[\chi\leq (1-\delta )\mu ]\leq \exp\left(-\frac {\delta ^{2}\mu }{2}\right),\;\; 0\leq \delta \leq 1
		\;\;\;\text{and}\;\;\;
		\Pr[\chi\geq (1+\delta )\mu ]\leq \exp\left(-\frac {\delta ^{2}\mu }{2+\delta }\right),\;\; 0\leq \delta\;.$$
	\end{thm}

The layering of vertices, defined next, is essentially the same as what was defined in the introduction, except that we place in $L_0$ all vertices with degree at most $\alphaX$, and not only all those with degree at most $3\alpha$.
		\begin{dfn}[Layering of $V$] \label{def:layering}  For a graph $G=(V,E)$ we define an \emph{$\alpha$-layering} of $G$ as follows. 
			$$\; L^{\alpha}_0(G)=\{v\in V \;:\; d(v)\leq \alphaX\},
			$$
			and for every $i>0$, 
			let
			\[ L^{\alpha}_i(G)=\left\{v\in V\;:\; v \notin \medcup_{j<i} L^{\alpha}_{j}(G) \text{ and }\left|\Gamma(v)\cap \medcup_{j<i} L^\alpha_{j}(G)\right|\geq d(v)-3\alpha\right\}.
			\]
Whenever the graph $G$ and the arboricity parameter $\alpha$ are clear from the context, we shall simply use $L_i$ instead of $L^\alpha_i(G)$.
We use $L_{<i}$ as a shorthand for $\bigcup_{j<i}L_j$, and $L_{\geq i}$ is defined analogously.
	\end{dfn}
Definition~\ref{def:layering} defines an iterative \emph{peeling} process for constructing the layers $L_0,L_1,\dots$.
As stated in the next theorem, if $\arb(G)\leq \alpha$, this process ends after at most $\log n$ iterations (with each vertex $v$ being placed in some layer $L_i$).

	\begin{thm}[\cite{BE}, Theorem 3.5, restated]
Let $G=(V,E)$ be a graph for which $\arb(G)\leq \alpha$. Then $V=\bigcup_{i=0}^{\ell}L_i$ for $\ell=\log_{3/2} n$, where $L_i$ is as defined in Definition~\ref{def:layering}.
	\end{thm}

\begin{dfn}
The {\sf degeneracy} of a graph $G=(V,E)$, denoted $\degen(G)$, is the maximum over all subgraphs of $G$, of the minimum degree in the subgraph. That is,
$\degen(G)=\max_{S\subseteq V}\{  \min_{v \in S}\{d_S(v)\}\}$.
\end{dfn}
	The following  is a well known relation between the arboricity of a graph, and its degeneracy.

\begin{thm}[Arboricity and degeneracy relation, e.g., ~\cite{eppstein1994arboricity, BE}] \label{thm:degen_rel}
	\[
	\arb(G)\leq \degen(G)\leq 2\arb(G).
	\]
\end{thm}
\begin{cor}\label{cor:A-G}
If $\arb(G) \geq \beta$, then $G$ contains a subset of vertices $R^\beta(G)$ of size at last $\beta$, such that
$d_{R^\beta(G)}(v) \geq \beta$ for every $v\in R^\beta(G)$.
\end{cor}

\section{The algorithm}\label{sec:alg}

In this section we present
the procedures that are the building blocks of our approximation algorithm.
The main procedure  distinguishes between graphs with arboricity at most $\alpha$ and graphs with arboricity greater than $100\alpha\log^2 n$. We first show how this can be performed with small ($O(1/\log n)$)  error probability.  The pseudo-code  appears in the  procedures \peel\ and \peelv\ (see Figures~\ref{fig:peel} and~\ref{fig:peelv}, respectively).
In Section~\ref{sec:search} we  reduce the error probability to $1/\poly(n)$,
and then show how the resulting algorithm can  be used to approximate the arboricity of a given graph, up to a  factor of $200\log^2 n$. The relevant pseudo-code  appears in Procedures \reduce\ and \est\ (see Figures~\ref{fig:reduce} and~\ref{fig:est}).

In the introduction we presented a high-level ideas behind the  algorithm and  its analysis. Here we provide the full details, while referring to some notions that were introduced in the introduction (and in particular the notions of peeled, pruned and active vertices).
For the ease of readability, we start by giving a verbal description of the procedures \peel\ and \peelv, followed by a ``road map'' of their analysis. (The procedures \reduce\ and \est\ and their analysis are fairly standard.)

From this point on, unless there is any ambiguity, whenever we refer to a graph $G$ such that $\arb(G)\leq \alpha$, we shall use the shorthand $L_i$ for $L_i^\alpha(G)$ (as defined in Definition~\ref{def:layering}), and whenever we refer to a graph $G$ such that $\arb(G) >  \beta$ for $\beta = \alphaX$, we shall use the shorthand $R$ for $R^{\beta}(G)$ (as defined in Corollary~\ref{cor:A-G}).

\subsection{The procedures \textsf{Peel} and \textsf{Peel-Vertex} }\label{subsec:procedures-verbal}
The procedure \peel\ is given query access to a graph $G$ and an arboricity parameter $\alpha$.
It starts by selecting a sample of (roughly $n/\alpha$) vertices, denoted $X_0$.
It then works in $\ell+1$ iterations, starting with $j=0$, where in iteration $j$ it peels a subset of the yet un-peeled sampled vertices, denoted $X_{j-1}$ (roughly speaking,  it peels those sampled vertices that belong to $L_j$.
The procedure \peel\ also keeps track
of the total number of queries performed. 
The peeling of a vertex, performed by the procedure \peelv, is done by implementing what we referred to in the introduction as the \emph{modified approximate peeling procedure}, and we discuss this further shortly.
If there are no remaining un-peeled sampled vertices after iteration $j=\ell$ (and the total number of queries did not exceed a certain threshold before reaching $j=\ell$),
then \peel\ returns $\YES$. Otherwise, it returns $\NO$.

The procedure \peelv, which is called on a vertex $v$ and a parameter $j$, maintains several data structures that contain information obtained regarding the original sampled vertices (i.e., those belonging to the set $X_0$) as well as additional vertices that are encountered in its (recursive) invocations (e.g., neighbors of vertices in $X_0$). 
We next elaborate on how \peelv\ works.

When $j=0$, \peelv$(v,j)$  queries the degree of $v$ and decides whether $v\in L_0$ based on the outcome $d(v)$. The procedure also increases the accumulated number of queries, $Q$, by $1$.
If $v$ does not belong to $L_0$, so that it is not peeled, then
the procedure  sets $\tp_1(v) = d(v)/(6\alpha)$ (if $v$ is peeled, then it sets $\tp_1(v)=0$).
This is the $1$-level cost associated with $v$ (that is,  the number of (neighbor) queries that will be performed on $v$ if \peelv$(v,j)$ is invoked with $j=1$).

Indeed, if \peelv$(v,j)$  is invoked with $j=1$, then the procedure performs $d(v)/(6\alpha)$ random neighbor queries
and lets the resulting (multi-)set of neighbors be denoted $S(v)$. The procedure also increases the accumulating number of queries, $Q$,  by $|S(v)|=d(v)/(6\alpha)$, and sets $\Nei{0}(v) = S(v)$. This is the initial set of active neighbors of $v$ (before any are peeled or pruned).

For any $j\geq 1$, the procedure proceeds as follows.
It recursively calls \peelv$(u,j-1)$  on each neighbor $u\in \Nei{j-1}(v)$, where  $\Nei{j-1}(v)\subseteq S(v)$ is the set of active neighbors of $v$ determined in the course of the invocations of \peelv$(v,j')$ for $j'=0,\dots,j-1$.
Each such call determines whether $u$ is peeled (in $j-1$ levels of recursion), or remains active.
In the latter case, $\tp_{j}(u)$, which was computed in the invocation of \peelv$(u,j-1)$,
holds the $j$-level cost associated with $u$.
In the former case,
$\tp_{j}(u)=0$.

The procedure \peelv\ then considers those neighbors
$u \in \Nei{j-1}(v)$ that remained active (following the recursive call to \peelv$(u,j-1)$).
It orders them according to their $j$-level cost $\tp_j(\cdot)$, and those $4\log n$ with the highest cost  are pruned. The updated set of neighbors $\Nei{j}(v)$ consists of those vertices in $\Nei{j-1}(v)$ that were neither peeled nor pruned.
If the size of $\Nei{j}(v)$ is sufficiently small (at most $8\log^2 n-j\cdot 4\log n$), then $v$ is peeled, which is indicated by setting $\tp_{j+1}(v)=0$. Otherwise,
$\tp_{j+1}(v)$ is set to be the sum, taken over all $u \in \Nei{j}(v)$, of $\tp_j(u)$.

\subsection{A road-map of the analysis}\label{subsec:road-map}
As described above, in the procedures \peel\ and \peelv, randomization comes into play in two ways.
The first is the choice of the initial set of random vertices, $X_0$ (selected by \peel). The size of $X_0$ is such that if  $\arb(G)> \alphaX$, then with high probability, $X_0$ will contain at least one vertex in $R$.
The second 
is the choice of the random (multi-)sets of neighbors $S(v)$, selected by \peelv\ (when invoked on a vertex $v$  with the parameter $j=1$).
The latter sets are selected for the vertices in $X_0$ as well as (some of) their descendants in the partial BFS trees that are constructed by \peelv. However,
for the sake of the analysis, it will be useful to consider, as a thought experiment, selecting the sets $S(v)$ for \emph{all} vertices $v\in V$, and establishing certain properties that hold with high probability over the choice of all these sets.
Note that once all these sets are selected, the execution of \peelv$(v,j)$ is determined for every vertex $v$ and parameter $j$.
In particular, it is determined for each vertex $v$ whether it is peeled by \peelv$(v,j)$, and if so, for which $j$.

The first building block of our analysis is Claim~\ref{clm:edge-samp}, which
 states useful properties of the sets, $\{S(v)\}_{v\in V}$ that hold with high probability.
Specifically, when $\arb(G)\leq \alpha$, we have that for every $i\in [0,\ell]$ and $v\in L_i $,  the number of neighbors in $S(v)$ that belong to layers $L_{\geq i}$  is not much larger than the expected value.
On the other hand, when $\arb(G)> \alphaX$, then for every vertex $v\in R$, the number of neighbors in $S(v)$ that belong to $R$ is not much smaller than the expected value.  When  that the sets $\{S(v)\}_{v\in V}$ have the aforementioned properties, we denote this event by $\mE$.

Based on Claim~\ref{clm:edge-samp} we show (in Claims~\ref{clm:pkv} and~\ref{clm:no_peel}, respectively),
that the following holds conditioned on the event $\mE$.
If $\arb(G)\leq \alpha$, then for every $v\in V$ there is some $j\in [\ell]$ such that $v$ is peeled by \peelv$(v,j)$,
while if $\arb(G)> \alphaX$, then for every $v\in R$, there is no $j\in [\ell]$ such that $v$ is peeled by \peelv$(v,j)$.
These two claims are then used to establish the correctness of \peel\
in the case that the total number of  queries performed does not exceed the allowed upper bound set by \peel (in addition to the condition that $\mE$ holds) --  see Claims~\ref{clm:YES} and~\ref{clm:NO}.

The main thrust of the analysis is showing that if $\arb(G)\leq \alpha$ and the event $\mE$ holds, then with sufficiently high probability, the total number of queries indeed does not exceed the allowed upper bound. To this end we define an imaginary ``wishful-thinking'' procedure, which we refer to as the \emph{downward-peeling procedure}.
This procedure is similar to \peelv, except that instead of pruning costly neighbors of a given vertex $v$, it prunes all neighbors of $v$ that belong to higher layers. Namely, if $v\in L_i$ (for $i\in [1,\ell]$), then it prunes every sampled neighbor $u$ in $S(v)$ that belongs to $L_{\geq i}$.

We then prove two central claims.
The first (Claim~\ref{clm:p_less_checkp}) is that if $\arb(G)\leq \alpha$, then
conditioned on the event $\mE$,
for every vertex $v\in V$, and $j\in[\ell]$, the number of queries performed in the course of the execution of \peelv$(v,j)$ is upper bounded by the number of queries performed by the downward peeling process on $v$ and $j$
(when the same sets of sampled neighbors are used).
The second (Claim~\ref{clm:bound_exp}) states that if $\arb(G)\leq \alpha$ and we invoke the downward peeling process on \emph{all} vertices $v\in V$ and all $j\in [\ell]$, then conditioned on the event $\mE$,
the expected query cost (over the choice of the vertex $v$) of invoking \peelv$(v,j)$ for $v$ and all $j\in [\ell]$ is $O(1)$.

By combining all aforementioned claims, we get (see Claim~\ref{clm:correctness}) that \peel$(G,\alpha)$ distinguishes between the case that $\arb(G)\leq \alpha$ and the case that $\arb(G)> \alphaX$ with sufficiently high probability.

\newcommand{\setT}{10n/(\alpha\log n)}
\newcommand{\UBq}{400t}

\algo{	
	{\bf Peel$(G,\alpha)$} \label{peel}
	\smallskip
	\begin{enumerate}
			\item Set $Q=0$. \label{step:loop_queries}
			\item Sample $t=\setT$ vertices uniformly, independently at random and denote the set of sampled vertices by $X_0$. \label{step:sampleX0}
			\item For $j=0$ to $\ell$ do:
			\begin{enumerate}		
				\item Initialize $X_{j+1}=\emptyset$
				\item For each $x\in X_j$ do:
				\begin{enumerate}		
					\item If $j\geq 1$ and $Q+\tp_j(x)$	exceeds $\UBq$, then \textbf{Return} $\NO$.
					\item Invoke \peelv$(x,j)$.
					\item If $\tp_{j+1}(x)\neq 0$, then let $X_{j+1}=X_{j+1} \cup  \{v\}$. \label{step:Rj}
				\end{enumerate}
			\end{enumerate}
			\item If $X_{\ell+1} = \emptyset$, then \textbf{Return} $\YES$. 
                                   \label{step:peel_return}
		\item Else \textbf{Return} $\NO$.\vspace{.5em} 
	\end{enumerate}
}
{The procedure \peel\ gets as input a value $\alpha$ and query access to a graph $G$. It  distinguishes between the case that $\arb(G)\leq \alpha$ and $\arb(G)>\alphaX$.}{fig:peel}

	\newcommand{\setp}{\frac{1}{6\alpha}}

\algo{
	{\bf Peel-Vertex$(v,j)$} \label{peelv}
	\smallskip
	\begin{enumerate}
		\item If \peelv$(v,j)$ was already called, then \textbf{Return}.
		\item If $j=0$:
		\begin{enumerate}
			\item \textbf{Query} $d(v)$ \label{step:query_deg} and update $Q=Q+1$.
			\item Set $\tp_{0}(v)=1$. \label{step:p0=1}
			\item If $d(v)\leq \alphaX$:\label{step:if_dv}  \comment{if $v\in L_0$, peel it} \\
                  $\;\;\;\;$ set $\tp_{1}(v)=0$ and  $\Nei{1}(v)=\emptyset$.
                              and  \textbf{Return} \label{step:zero_p1} 
             \item Else set $\tp_{1}(v)=d(v)/(6\alpha)$ \label{step:p1v_Rv} and \textbf{Return} \comment{o.w., set its $1$-cost}
		\end{enumerate}
		\item If $j=1$:
		\begin{enumerate}
            \item Select u.i.r. $d(v)/(6\alpha)$ indices in $[d(v)]$
              and perform a neighbor query on $v$ and each selected index. Denote the resulting (multi-)set of neighbors by $S(v)$. \comment{sample nbrs}
		 \label{step:set_tGl0}
		 \item Update $Q=Q+|S(v)|$. \label{step:update_qt_Sv}
			\item Let $\Nei{0}(v)=S(v)$. \label{step:Gam_zero}  \comment{initial set of active neighbors}
		\end{enumerate}
		\item For each $u\in \Nei{j-1}(v)$ do:\label{step:loop_recursion} \comment{for each remaining active nbr do}
		\begin{enumerate}
			\item Invoke  \peelv$(u,j-1)$,\label{step:invoke_rec} \comment{recursively invoke the procedure}
    		\item If $\tp_j(u)=0$, then place $u$ in  $\tP_j(v)$.  \label{step:def_Pj} \comment{$u$ should be peeled}
		\end{enumerate}
        \item Let $\tNei{j}(v) = \Nei{j-1}(v) \setminus \tP_j(v)$  \label{step:update_tGj} \comment{remove (newly) peeled neighbors}
		\item Let $\tH_{j}(v)$ be the set of $\min\{4\log n,|\tNei{j}(v)|\}$
         vertices in
          $\tNei{j}$ with highest  $\tp_{j}$ value
             \label{step:Hj}
		\item Let
       $ \Nei{j}(v)=\tNei{j}(v)\setminus \tH_{j}(v))$. \comment{Remove pruned (costly) nbrs}
     \label{step:update_tGj_low}
		\item If  $|\Nei{j}(v)| \leq \tau(j)=8\log^2 n-j\cdot 4\log n$:    \comment{peel $v$}\\
             $\;\;\;$ set $\tp_{j+1}(v)=0$ and  $\Nei{j+1}(v)=\emptyset$.
                       \label{step:deactivate}
		\item Else:
              $\tp_{j+1}(v)=\sum_{u \in \Nei{j}(v)} \tp_{j}(u)$. \comment{compute $(j+1)$-cost}
                  \label{step:pj}
	\end{enumerate}
}
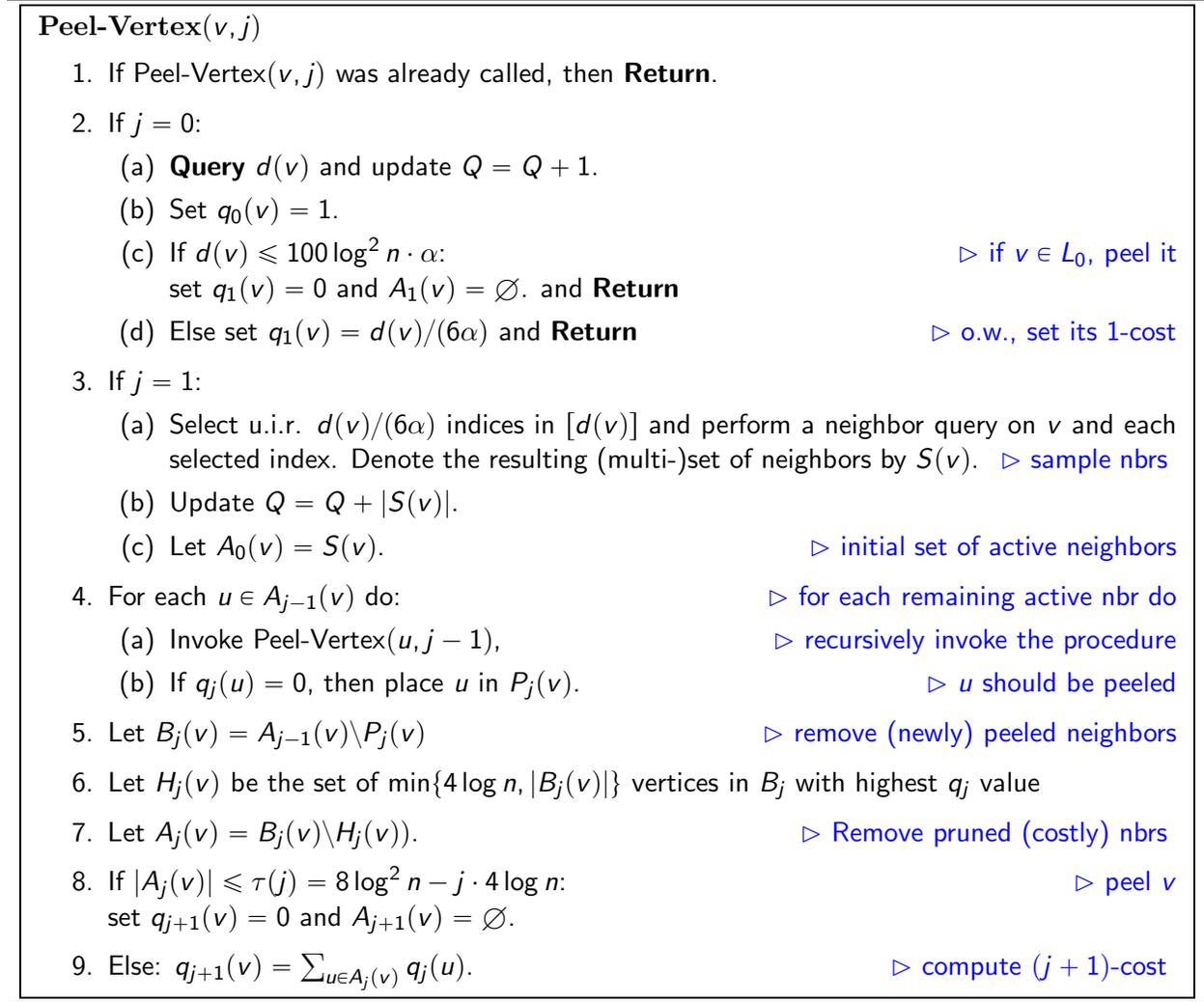
{The procedure gets as input a vertex $v$ and an index $j$. It decides whether $v$ should be peeled, and  if not,  it computes $v$'s updated set of active neighbors, and its $(j+1)$-cost. }{fig:peelv}

\subsection{Correctness}\label{subsec:cor}
In this subsection we prove the correctness of the procedure \peel\ conditioned on the number of queries $Q$ not exceeding the allowed upper bound. In Section~\ref{subsec:cost} we  bound the probability that $Q$ exceeds this bound when $\arg(G)\leq \alpha$.
We start by proving the following claim regarding the random sets of neighbors $S(v)$.

\begin{clm}\label{clm:edge-samp}
	Consider (as a thought experiment) sampling a (multi-)set of neighbors $S(v)$ for {\sf every } $v\in V$, by
performing $d(v)/(6\alpha)$ independent random neighbor queries.
	Then the following hold.	
	\begin{itemize}
		\item If $\arb(G) \leq \alpha$, then with probability at least $1-1/n$, for every $v\in V$, if $v\in L_i$, then      
$|S(v)\cap L_{\geq i}| \leq 3\log n$.
		\item If $\arb(G)> \alphaX$, then with probability at least $1-1/n^4$, for every $v\in R$,  $|S(v)\cap R| \geq 8\log^2 n$. 
	\end{itemize}
\end{clm}	

\begin{proof}
Consider first the case that $\arb(G) \leq \alpha$.
	Fix a vertex $v$ and let $s=d(v)/(6\alpha)$. For each $r = 1,\dots,s$ let $\chi_r$ be a Bernoulli random variable whose value is $1$ if the $r\th$ sampled neighbor of $v$ belongs to $L_{\geq i}$.
Let $\chi = \sum_{r=1}^s \chi_t$, so that $\EX[\chi] = 1/2$ and
  $|S(v) \cap L_{\geq i}| = \chi$  (recall that $S(v)$ is a multi-set and hence when we consider its intersection with $L_{\geq i}$ we obtain a multi-set).
	By (the second item of) the multiplicative Chernoff bound (Theorem~\ref{thm:chernoff}), for $\delta=4\log n$,
	\[
	\Pr[\chi >(1+4\log n)\cdot 1/2]\leq
	\exp\left(-\frac{\delta^2 \cdot 1/2}{2+\delta}\right)\leq
	\exp\left(-\frac{16\log^2 n /2}{2+4\log n}\right)\leq \frac{1}{n^{2}}.
	\]
	Therefore, with probability at least $1-1/n^{2}$, $|S(v) \cap L_{\geq i}|\leq 3\log n$. The first item of the claim follows by taking a union bound over all vertices in $G$.

	Now consider the case that $\arb(G)>\alphaX$, and let $v$ be some vertex in $R$.	
	By 
Corollary~\ref{cor:A-G},
$|\Gamma(v)\cap R|\geq \alphaX$, implying that  $\EX[|S(v)\cap R|]\geq \alphaX \cdot \frac{1}{6\alpha}\geq 16\log^2 n$
(once again, recall that $S(v)$ is a multi-set, and the same holds for $S(v)\cap R$).
	Therefore, 	by (the first item of)  the multiplicative Chernoff bound (Theorem~\ref{thm:chernoff}),

	\[
	\Pr[|S(v)\cap R| <(1-1/2)\cdot 16\log^2 n]\leq \exp\left(16\log^2 n/8\right)\leq \frac{1}{n^{5}}\;,
	\]
	Hence, for a fixed $v$, with probability at least $1-1/n^5$, $|S(v)\cap R|\geq 8\log^2 n$. The second item of the claim follows by taking a union bound over all vertices in $R$.
\end{proof}	

\vspace{1ex}

\begin{dfn}[Successful neighbor sampling]\label{def:succ-edges}
We refer to an event where the relevant item in Claim~\ref{clm:edge-samp} holds
(i.e., the first item if $\arb(G)\leq \alpha$ and the second item if $\arb(G)>\alphaX$)
as {\sf success of the neighbor sampling process},
and denote this event by $\mE$.
\end{dfn}


Consider (again as a thought experiment) running $\peel(v,j)$ on all vertices $v\in V$ for $j=0$ to at most  $\ell$.
For every fixed choice of $S(V)= \{S(v)\}_{v \in V}$,
all these executions  are deterministic.
In what follows we analyze the correctness and expected query complexity of these invocations (that is, for now we assume that \peelv\ is invoked for \emph{all}  vertices), where the probability  is taken over the choice of $S(V)$.

We first introduce the following notation for the sets of vertices that are peeled in the different iterations $j$.

\begin{ntn}\label{ntn:Pj}
	For each $j\in[\ell]$, let
	\[\tP_j = \{v\;:\; \tp_j(v) > 0 \mbox{ and } \tp_{j+1}=0\}\;,\]	
denote the set of vertices that are peeled when $\peelv(v,j)$ is invoked,
and let $\tP_{\leq j}=\medcup_{j'\leq j}\tP_{j'}$.
\end{ntn}
The next observation follows directly from the description of \peelv.
\begin{obs}\label{obs:Pj}
	For each $j\in[\ell]$ and for $\Nei{j}(\cdot)$ and $\tau(j)$ as defined in \peelv,
	\[
	\tP_j=
	\begin{cases}
		L_0 & j=0\\
		\left\{v\;:\;v\notin \tP_{\leq j-1} \; \;\& \;\;|\Nei{j}(v)|\leq \tau(j)\right\} & j\in [1,\ell]
	\end{cases}
	\]
\end{obs}

\subsubsection{The case $\arb(G)\leq \alpha$.}

We  prove that, conditioned on the event $\mE$, for every $v\in V$, if $v\in L_i$, then $v$ is peeled by
$\peelv(v,j)$ for some $j \leq i$.

\begin{clm}\label{clm:pkv}
	Let $G$ be a graph for which $\arb(G) \leq \alpha$ and assume that event $\mE$ holds.
	For every $i\in[\ell]$ and $v\in L_i$, we have that $v\in \tP_{\leq i}$.
\end{clm}
\begin{proof}
	We prove the claim by induction on $i$.
	For $i=0$, $v\in L_0$, and it holds by Step~\ref{step:zero_p1} that	$\tp_1(v)=0$,
so that $v\in \tP_0$.

	We now assume that the claim holds for all $i'\leq i-1$, and prove it for $i$.
If an invocation of \peelv$(v,i')$ for $i' < i$ already set  $\tp_{i'+1}(v)$ to $0$, then $v\in P_{i'}\subseteq P_{\leq i}$, and we are done.
 Otherwise, 	consider the invocation of \peelv$(v,i)$.
	By 	the induction hypothesis, for every $u\in \Nei{i-1}(v)\cap L_{<i}$, $u\in P_{<i}$.
	Therefore,
  $\tNei{i}(v) \cap L_{<i}=\emptyset$, so that $\tNei{i}(v)\subseteq L_{\geq i}$.
	Together with the fact that $\tNei{j}(v)\subseteq S(v)$ for every $j$, we get that $\tNei{i}(v)\subseteq (S(v)\cap L_{\geq i})$. By the definition of the event $\mE$, it holds that $|S(v)\cap L_{\geq i}|\leq 3\log n$,
	  and so $|\tNei{i}(v)| \leq 3\log n$. Hence, due to Step~\ref{step:update_tGj_low},
	  $|\Nei{i}(v)|=0$, and so by Step~\ref{step:deactivate}, $\tp_{i+1}(v)=0$, implying that $v\in P_{\leq i}$.
\end{proof}

We use Claim~\ref{clm:pkv} to prove the next claim.
\begin{clm}\label{clm:YES}
	Let $G$ be a graph for which $\arb(G) \leq \alpha$, and assume that the event $\mE$ holds.
	If $Q$ does not exceed $\UBq$, then \peel$(G,\alpha)$ returns $\YES$.
\end{clm}
\begin{proof}
	Since $\arb(G) \leq \alpha$, and by the assumption that event $\mE$ holds,  by Claim~\ref{clm:pkv}, for every $v\in L_i$,  $v \in P_{\leq i}$.
	Hence, by Step~\ref{step:Rj}, for every $j\in [\ell+1]$, $X_{j}\cap L_{< j}=\emptyset$.
	Since $\arb(G) \leq \alpha$, every $v$ is in $L_i$ for some $i\in [\ell]$, and therefore, $X_{\ell+1}=\emptyset$, and if the algorithm reaches Step~\ref{step:peel_return}, then it returns $\YES$.
	Hence, if $Q$ does not exceed $\UBq$ (causing the \peel$(G,\alpha)$ to abort and return $\NO$), then \peel$(G,\alpha)$ returns $\YES$.
\end{proof}

\subsubsection{The case $\arb(G) > \alphaX$}

\begin{clm}\label{clm:no_peel}
	Let $G$ be a graph for which $\arb(G) >  \alphaX$,  and assume that event $\mE$ holds.	Then  for every $v\in R$, $v \notin P_{\leq \ell}$.
\end{clm}
\begin{proof}
	 We  prove the claim by induction on $j$.  We shall actually prove a slightly stronger claim: that for every $j\in[\ell]$, $\tp_{j+1}(v)>0$ and  $|\Nei{j}(v)\cap R|\geq 8\log ^2 n-j\cdot 4\log n$.
	Fix a vertex $v\in R$.
	For every $v\in R$,
	 $d(v)\geq |\Gamma(v)\cap R|>\alphaX$,
	the condition in Step~\ref{step:if_dv} in \peelv\ does not hold,
	 and therefore $\tp_1(v)=d(v)/(6\alpha)>0$ (see Step~\ref{step:p1v_Rv}).	
By the assumption that the event $\mE$ holds, and by the second item in Claim~\ref{clm:edge-samp}, for every $v\in R$, $|S(v)\cap R|\geq 8\log^2 n$.
	Since $|\Nei{0}(v)|=|S(v)|$, it holds that $|\Nei{0}(v)|\geq 8\log^2 n$, as required.
	
	Now assume the claim holds for $j-1$, and we  prove it for $j$. By the induction hypothesis, $|\Nei{j-1}(v)\cap R|\geq 8\log^2 n-(j-1)\cdot 4\log n$.
	For every $u\in \Nei{j-1}(v)\cap R$, since in particular $u\in R$, by the induction hypothesis,  $\tp_j(u)>0$,  and  $|\Nei{j-1}(u)\cap R|\geq 8\log^2 n-(j-1)\cdot 4\log n>0$ (where the last inequality is since $j\leq \log n$).
	Hence,
	$\left(\Nei{j-1}(v)\medcap R\right) \medcap \tP_j(v) =\emptyset$,
	so that
	$\Nei{j-1}(v)\cap R\subseteq \tNei{j}(v)$, and it follows that
	$\Nei{j}(v)=
	\tNei{j}(v)\setminus \tH_j(v)
	\supseteq \left(\Nei{j-1}(v)\medcap R\right)\setminus H_j(v)$.
Since $|\tH_j(v)|\leq 4\log n$,
	 \[|\Nei{j}(v)\cap R|\geq |\Nei{j-1}(v)\cap R|-4\log n\geq 8\log^2 n-j\cdot 4\log n>0.\]
	Therefore, $|\Nei{j}(v)|>0$, and since  for every $u\in \Nei{j}(v)$, $\tp_j(u)>0$ (as mentioned earlier, this is due to Steps~\ref{step:def_Pj} and~\ref{step:update_tGj}), it follows that
	 $\tp_{j}(v)>0$, so that  the induction claim holds. Hence, for every $j\in[\ell]$, $q_{j+1}(v)>0$, and so $v\notin P_{\leq \ell}$.
\end{proof}

We next lower bound the probability that \peel$(G,\alpha)$ returns $\NO$ when $\arb(G) > \alphaX$.
\begin{clm}\label{clm:NO}
	Let $G$ be a graph for which $\;\arb(G) > \alphaX$, and assume that event $\mE$ holds.
	Then with probability at least $1-1/n^4$, 	\peel$(G,\alpha)$ returns $\NO$.
\end{clm}
\begin{proof}
	First we argue that
with high probability, $X_0\cap R\neq \emptyset$.
	In a single vertex sampling attempt,
	the probability that the vertex  chosen to $X_0$ is not in $R$ is $1-|R|/n$. Hence, the probability that in $\setT$ attempts no vertex of $R$ is chosen to $X_0$ is $(1-|R|/n)^{\setT}<(1-\alphaX/n)^{\setT}<1/n^4$. Condition on this event.

		By Claim~\ref{clm:no_peel}, conditioned on the event $\mE$, for every $u\in S(v)$, if $u\in R$, then $\tp_{j}(u)\neq 0$ for every $j\in [1,\ell+1]$. Since $X_0(v)\cap R\neq \emptyset$, it holds that $X_{\ell}\neq \emptyset$.
	Therefore, conditioned on the event $\mE$, with probability at least $1-1/n^4$,  \peel$(G,\alpha)$ returns $\NO$.
\end{proof}

\subsection{Bounding the query complexity}\label{subsec:cost}

Recall that we are still within the thought experiment by which all neighbor (multi-)sets $S(v)$ were selected in advance, and we invoke \peelv\ on every $v\in V$ for $j=[\ell]$ (more precisely, once \peelv$(v,j)$ peels $v$, i.e., sets $\tp_k(v)=0$ for every $k\in [j+1,\ell]$, then no further invocations of \peelv$(v,j')$ for $j'>j$ are performed).

For the sake of the analysis, it will be convenient to define
the values  $\tp_j(v)$ and sets $\Nei{j}(v)$ for vertices that were already peeled in previous iterations.

\begin{dfn}
For a vertex $v\in P_{j}$, we let $\tp_{k}(v)=0$ and $\Nei{k}(v)=\emptyset$ for all $k\in[j+2].$ (Note that the index $k$ goes from $j+2$ to $\ell$, since if $v\in P_j$, $\tp_{j+1}(v)$ and $\Nei{j+1}(v)$ are already defined.)
\end{dfn}

\begin{clm}\label{clm:ub_cost}
	The number of queries performed during the execution of \peelv$(v,j)$ (if invoked)
	is at most $2\tp_{j}(v).$
\end{clm}

\begin{proof}
	We prove the claim by induction on $j$, starting with $j=0$ and  \peelv$(v,0)$.
	The query complexity is due to the degree query in  Step~\ref{step:query_deg}, and is hence $1$.
	By Step~\ref{step:p0=1}, $\tp_0(v)$ is set to 1.
	
	For $j=1$, the query complexity of \peelv$(v,1)$ is due to the neighbor queries in
	Steps~\ref{step:set_tGl0} and the recursive invocations on the sampled neighbors in  Step~\ref{step:invoke_rec}.
	The query complexity of Step~\ref{step:set_tGl0} is $d(v)/(6\alpha)$ and the equality  $\tp_{1}(v)=d(v)/(6\alpha)$ is by Step~\ref{step:p1v_Rv}.
	For each $u\in \Nei{0}(v)$ (where recall that $|\Nei{0}(v)|=|S(v)|=d(v)/(6\alpha)$), \peelv$(u,0)$ is invoked.
Since for every vertex $u$, the query complexity of \peelv$(u,0)$ is $1$,  the query complexity of the recursive invocations is $d(v)/(6\alpha)$. Hence, the  query complexity   \peelv$(v,1)$ is  $\tp_1(v)+d(v)/(6\alpha)=2\tp_1(v)$.

	For the induction step, assume that the claim holds for  $j-1$, and we shall prove it holds for $j$.
	The only queries performed for $j> 1$ are due to the recursive invocations \peelv$(u,j-1)$ for every $u\in \Nei{j-1}$ in
 Step~\ref{step:invoke_rec}.	By the induction hypothesis, for every $u\in \Nei{j-1}(v)$, the query complexity of \peelv$(u,j-1)$ is at most  $2\tp_{j-1}(u)$.
	Hence, the query complexity of \peelv$(v,j)$ is
	$\sum_{u \in \Nei{j-1}(v)} 2\tp_{j-1}(u)=2\tp_{j}(v)$, where the equality is by the setting of $\tp_j(v)$ in
	Step~\ref{step:pj} during  the invocation of \peelv$(v,j-1)$.
\end{proof}

Since the 	query complexity of \peelv$(v,j)$ is bounded by  $2\tp_j(v)$, we would like to bound the expected value of $\tp_j(v)$. To this end we compare the process of \peelv\ with the following ``wishful-thinking'' process that was mentioned in the introduction.

\begin{dfn}[Downward peeling procedure] \label{def:downward_peel}  \label{def:pcheck}
	The \emph{downward peeling procedure} is identical to \peelv, except that in an invocation on any vertex $v\in L_i$, all of the ``upward'' sampled neighbors of $v$ are pruned (i.e., its neighbors in layers $L_{\geq i}$), rather than the costly ones.
	(To be precise, once the upward pruning is performed the first time, for $j=1$, no upward neighbors remain in the set of sampled, and therefore no more pruning is performed.)

We denote by $\cNei{j}(v)$
and $\ap_j(v)$  the sets and costs in the downward peeling procedure that are analogous  to $\Nei{j}(v)$ and $\tp_j(v)$, respectively, from the procedure \peelv.

\end{dfn}

In order to bound the expected complexity of our peeling procedure, we first prove that it is bounded by the complexity of the downward peeling procedure (for the same choice of $S(V) = \{S(v)\}_{v\in V}$\;), and then continue to bound the expected complexity of the latter.
Analogously to Notation~\ref{ntn:Pj} and Observation~\ref{obs:Pj}:

\begin{ntn}\label{ntn:cPj}
	For each $j=[0,\ell]$, let
	\[\cP_j = \{v\;:\; \ap_j(v) > 0 \mbox{ and } \ap_{j+1}=0\}\;\]
	and let $\cP_{\leq j}=\medcup_{j'\leq j}\cP_{j'}$.
\end{ntn}
\begin{obs}\label{obs:cPj}
	For each $j\in[0,\ell]$ and for $\cNei{j}(\cdot)$ as defined for the downward peeling procedure (and $\tau(j)$ as defined in \peelv),
	\[
	\cP_j=
	\begin{cases}
		L_0 & j=0\\
		\left\{v\;:\;v\notin \cP_{\leq j-1} \; \;\& \;\;|\cNei{j}(v)|\leq \tau(j)\right\} & j\in [1,\ell]
	\end{cases}
	\]
\end{obs}

We prove the following relations (where $\tNei{j}(v)$ is as defined in Step~\ref{step:update_tGj} of \peel$(v,j)$.
\begin{clm}\label{clm:Gamj-tGamj}
	Conditioned on the event $\mE$, for every $i,j\in[\ell]$ and every $v\in L_i\setminus \tP_{\leq j}$, 	\[ |\Nei{j}(v)|\leq |\tNei{j}(v)\medcap L_{<i}|
	\;\;\;\;\text{and} \;\;
	\sum_{u\in \Nei{j}(v)}\tp_{j}(u)	\leq
	\sum_{u\in \tNei{j}(v)\cap L_{<i}}\tp_{j}(u)\;.
	\]
\end{clm}
\begin{proof}
	Observe that for any $j$, $\tNei{j}(v)\subseteq S(v)$. Hence, $(\tNei{j}(v)\medcap L_{\geq i} )\subseteq (S(v)\medcap L_{\geq i})$. By the conditioning on the event $\mE$, $|S(v)\medcap L_{\geq i}|\leq 3\log n$. Therefore,
	$|\tNei{j}(v)\medcap L_{<i}|\geq \max\{0,|\tNei{j}(v)|-3\log n\}$.
Since in the pruning process (Steps~\ref{step:Hj} and~\ref{step:update_tGj_low}), $\min\{4\log n, |\tNei{j}(v)|\}$ vertices are removed from $\tNei{j}(v)$, we have that $|\Nei{j}(v)| = \max\{0,|\tNei{j}(v)|-4\log n\}$.
It follows that $|\Nei{j}(v)|\leq |\tNei{j}(v)\medcap L_{<i}|.$
	
We now turn to the second part of the claim.
If $|\tNei{j}(v)|\leq 4\log n$,
then by Step~\ref{step:update_tGj}, $\Nei{j}(v)=\emptyset$ and the claim holds (since $\tp_j(u)\geq 0$ for every $u$ and $j$). Hence, assume that $|\tNei{j}(v)|> 4\log n$, which implies that $|\tH_j(v)|= 4\log n$ (where recall $\tH_j(v)$ is set in Step~\ref{step:Hj}).

	Since $\Nei{j}(v)=\tNei{j}(v)\setminus \tH_j(v)$,
	\begin{align*}
	\sum_{u\in \Nei{j}(v)} \tp_j(u)=
	\sum_{u\in \tNei{j}(v)} \tp_j(u)-	\sum_{u\in \tH_j(v)} \tp_j(u).
	\end{align*}
	Also,
	\begin{align*}
	\sum_{u\in \tNei{j}(v)\cap L_{<i}} \tp_j(u)=
\sum_{u\in \tNei{j}(v)} \tp_j(u)-	\sum_{u\in \tNei{j}(v)\cap L_{\geq i}} \tp_j(u).
\end{align*}
Recall that by Step~\ref{step:Hj}, $\tH_j(v)$ is the set of highest $\tp_j(u)$ values in $\tNei{j}(v)$.
This together with the fact that $|\tNei{j}(v)\medcap L_{\geq i}| \subseteq |S(v)\medcap L_{\geq i}| \leq 3\log n < 4\log n=|\tH_j(v)|$ implies that
\begin{align*}
\sum_{u\in \tH_j(v)} \tp_j(u) > \sum_{u\in \tNei{j}(v)\cap L_{\geq i}} \tp_j(u).
\end{align*}
Therefore,
\begin{align*}
	\sum_{u\in \Nei{j}(v)} \tp_j(u)<
		\sum_{u\in \tNei{j}(v)\cap L_{<i}} \tp_j(u),
\end{align*}
as claimed.
\end{proof}

Next we relate between the sets $\tP_{\leq j}$ and $\cP_{\leq j}$ and between
$\tNei{j}(v)$ and $\cNei{j}(v)$.
\begin{clm}\label{clm:contained}
	Conditioned on the event $\mE$,
	for every $j\in [0,\ell]$,
	$\cP_{\leq j}\subseteq \tP_{\leq j}$, and for every $i$ and  $v\in L_i$ and $j\in [0,\ell]$, $\tNei{j}(v)\medcap L_{< i} \subseteq \cNei{j}(v)$.
\end{clm}
\begin{proof}
	We prove the two parts of the claim by induction on $j$.
	By 
Observations~\ref{obs:Pj} and~\ref{obs:cPj},
	for $j=0$, $\tP_0=L_0=\cP_0$, and $\tNei{0}(v)=|S(v)|=\cNei{0}(v)$.
	For $j=1$, $\cNei{1}(v)=(S(v)\medcap L_{< i})\medcap L_{>0}$, and
	$\tNei{1}(v)\medcap L_{<i}=(S(v)\medcap L_{<i})\medcap L_{>0}$. Therefore, $\tNei{1}(v)\medcap L_{<i}=\cNei{1}(v)$.
	Also by the aforementioned observations, 
	$\tP_1 = \left\{v\;:\;v\notin \tP_0 \; \;\& \;\;|\Nei{1}(v)|\leq \tau(1)\right\}$, and
	$\cP_1 = \left\{v\;:\;v\notin \cP_0 \; \;\& \;\;|\cNei{1}(v)|\leq \tau(1)\right\}$.
	By Claim~\ref{clm:Gamj-tGamj}, $	|\Nei{1}(v)|\leq |\tNei{1}(v)\medcap L_{<i}|	$.
	Hence, for every $v$,  $|\Nei{1}(v)|\leq |\cNei{1}(v)|$, implying that
	$\tP_1\supseteq \cP_1.$
	
	For the induction step, we assume both parts of the claim hold for  $j-1\geq 1$ and prove each part for $j$.
     By Step~\ref{step:def_Pj},
	\[\tNei{j}(v) = (\Nei{j-1}(v)\setminus\tP_{j-1})) \subset (\tNei{j-1}(v)\setminus\tP_{j-1})\;\]
	Recall that by the definition of $\cNei{j}(v)$,  the pruning of the upward neighbors only happens once, for $j=1$. Therefore, for $j\geq 2$,
	\[  \cNei{j}(v)=\cNei{j-1}(v)\setminus \cP_{j-1} .\;\]
	By the induction hypothesis, $\tNei{j-1}(v) \medcap L_{< i} \subseteq \cNei{j-1}(v)$ and
	$\cP_{\leq j-1}\subseteq \tP_{\leq j-1}$, and hence
	$\tNei{j}(v) \medcap L_{<i} \subseteq \cNei{j}(v)$ follows.
	
	\sloppy
	By Observations~\ref{obs:Pj} and~\ref{obs:cPj},
	$\tP_j = \left\{v\;:\;v\notin \tP_{<j} \; \;\& \;\;|\Nei{j}(v)|\leq \tau(j)\right\}$, and	$\cP_j = \left\{v\;:\;v\notin \cP_{<j} \; \;\& \;\;|\cNei{j}(v)|\leq \tau(j)\right\}$.
	By Claim~\ref{clm:Gamj-tGamj}, $|\Nei{j}(v)|\leq |\tNei{j}(v)\medcap L_{<i}|$.
	Since we have just shown that $\tNei{j}(v) \medcap L_{<i} \subseteq \cNei{j}(v)$,
	we have that $|\Nei{j}(v)| \leq |\cNei{j}(v)|$.
	Now consider a vertex $v\in \cP_{\leq j}$. Then either $v\in \cP_{< j}$ and by the induction hypothesis, $v\in \tP_{<j}$, or $|\cNei{j}(v)|\leq \tau(j)$, in which case  $|\Nei{j}(v)|\leq \tau(j)$, and $v\in \tP_j$. This concludes the proof.
\end{proof}

We are now ready to prove that, conditioned on $\mE$,  the cost of \peelv$(v,j)$ is bounded by the cost of the downward procedure.

\begin{clm}\label{clm:p_less_checkp}
	Let $G$ be a graph for which $\arb(G)\leq \alpha$.
	Conditioned on the event $\mE$,
	for every $v\in V$, and $j\in[\ell+1]$,
	\[
	\tp_j(v)\leq \ap_j(v).
	\]
\end{clm}
\begin{proof}
	We shall prove the claim by induction on $j$.
	For $j=0$, $\tp_0(v)=\ap_0(v)=1$, and so the claim holds.
	For $j=1$,
	$
	\tp_1(v)=|S(v)|=\ap_1(v)\;.
	$
	Now assume the claim holds for every $1\leq j'\leq j-1$, and we prove it for
	$j$.
	First, if $v\in \cP_{\leq j}$, then by Claim~\ref{clm:contained}, $v\in \tP_{\leq j}$, implying that if $\ap^j(v)=0$, then so is $\tp^j(v)=0$.
	Otherwise, by Step~\ref{step:pj},  $\tp_j(v)=\sum_{u\in \Nei{j-1}(v)} \tp_{j-1}(u)$,
	and similarly,
	$\ap_j(v)=\sum_{u\in \cNei{j-1}(v)} \ap_{j-1}(u)$.
	Furthermore, by Claim~\ref{clm:Gamj-tGamj},
	\[\sum_{u\in \Nei{j-1}(v)}\tp_{j-1}(u)\leq
	\sum_{u\in \tNei{j-1}(v)\cap L_{<i}}\tp_{j-1}(u),\] and by Claim~\ref{clm:contained}
	\[(\tNei{j-1}(v)\medcap L_{<i}) \subseteq \cNei{j-1}(v).\] Putting everything together, we get
	\begin{align*}
		\tp_j(v) =  \sum_{u\in \Nei{j-1}(v)}\tp_{j-1}(u) \leq \sum_{u\in \tNei{j-1}(v)\cap L_{<i}}\tp_{j-1}(u)\leq
		\sum_{u\in \cNei{j-1}(v)}\tp_{j-1}(u)\leq
		\sum_{u\in \cNei{j-1}(v)}\ap_{j-1}(u)
		=\ap_j(v)\;.
	\end{align*}
	This completes the proof.
\end{proof}

\begin{ntn}\label{dfn:sigma}
	For $u\in L_k$ and $i\geq k$, let $\sigma_i(u)=|\{v\in L_i  \mid u\in S(v)\}|$. That is, $\sigma_i(u)$ is the number of vertices in layer $L_i$ that have chosen $u$ to their (multi-)set $S(v)$.
	For $u\in L_k$, $\sigma(u)=\sum_{i\geq k}\sigma_i(u)$.
\end{ntn}

Recall that in the downward peeling procedure,  for any vertex $w$,
if $w\in L_0$, then $\ap_1(w)=0$ and $\cNei{1}(w) = \emptyset$, and if $u\in L_i$ for $i>1$, then
$\ap_1(w) = d(w)/(6\alpha)$ and
$\cNei{1}(u) = S(u)\cap L_{<i}$. This implies that if  	we consider the partial BFS tree defined by the downward peeling procedure for a vertex $u\in L_k$ (the root of the tree) and index $j$ (the depth of the tree), then all vertices in the tree belong to $L_{<k} \cup \{u\}$. This in turn leads to the next observation.
\begin{obs} \label{obs:qju_depend}
	For every $k\in [0,\ell]$ and  $u\in L_k$, the following holds.
	For every $j\in[\ell+1]$, the value   $\ap_j(u)$ and the identity of  vertices in the set $\cNei{j}(u)$  only depend
on the choices of the sets $S(w)$ for $w\in L_{<k}\cup\{u\}$.
\end{obs}

\begin{clm} \label{clm:exp_indep}
	For every $k\in[\ell], u\in L_k$, for any $j\in[k,\ell]$,
	\[\EX\left[\sigma(u) \cdot \ap_{j}(u) \mid \mE \right]\leq  \EX[\sigma(u)\mid \mE]\cdot \EX[\ap_{j}(u)\mid \mE].\]
\end{clm}
\begin{proof}
	By Observation~\ref{obs:qju_depend},
	the value of $\ap_j(u)$ only depends on the choice of random neighbors of $u$ and of vertices $w$ in layers $L_i$ such that $i<k$. This is in contrast to $\sigma(u)$ that  depends on the 
choices of subsets $S(v)$ of neighbors $v$ of $u$ that belong to layers $L_i$ for $i\geq k$. Therefore, it  follows that $\ap_j(u)$ and $\sigma(u)$ depend on a disjoint sources of randomness. Hence, $\sigma(u)$ and $\ap_j(u)$ are independent, and this holds also in the case that the event $\mE$ occurs.
\end{proof}	

\begin{clm}\label{clm:exp_sigma}
	Let $G$ be a graph for which $\arb(G)\leq \alpha$. For every $u\in V$,
	$\EX[\sigma(u)]\leq 1/2$.
\end{clm}
\begin{proof}
	Since $\arb(G) \leq \alpha$,
	by the definition of the layers in Definition~\ref{def:layering},
	for every $v\in L_k$,
	$|\Gamma(v)\medcap L_{\geq k}|\leq 3\alpha$.
	By Step~\ref{step:p1v_Rv}, it holds that
	$\EX[\sigma(u)]=|(\Gamma(u)\medcap L_{\geq k})|/(6\alpha)\leq 1/2$.
\end{proof}

\begin{clm}\label{clm:bound_exp}
	Let $G$ be a graph for which $\arb(G)\leq \alpha$.
	Then
	\[
	\EX\left[\sum_{i=0}^\ell \sum_{v\in L_i} \sum_{j=0}^{\ell} \tp_j(v) \mathrel{\Big|} \mE\right]
	\leq 2n\;.
	\]
	
\end{clm}
\begin{proof}
	By Claim~\ref{clm:pkv}, conditioned on $\mE$, for every $v\in L_i$, $v \in \tP_{\leq i}$, and therefore $\tp_k(v)=0$ for every $k\in [i+1,\ell]$.
	Therefore, for $v\in L_i$, $\tp_j(v)\neq 0$ only for $j\in [0,i]$, or alternatively, fixing an index $j$,  $\tp_j(v)$ is only non-zero for vertices $v$ in layers $L_i$ through $L_{\ell}$.
	Hence,  conditioned on $\mE$,
	\begin{align}
		\sum_{i=0}^\ell \sum_{v\in L_i} \sum_{j=0}^{\ell} \tp_j(v) \;=\;
 \sum_{j=0}^\ell \sum_{i=j}^\ell \sum_{v\in L_i}  \tp_j(v)
		\leq \sum_{j=0}^\ell \sum_{i=j}^\ell \sum_{v\in L_i}  \ap_j(v),
	\end{align}	
	where the last inequality is due to Claim~\ref{clm:p_less_checkp}.
	Hence,
	\begin{align}\label{eqn:expP_vs_checkp}
		\EX\left[\sum_{j=0}^\ell  \sum_{i=j}^\ell  \sum_{v\in L_i} \tp_j(v) \mathrel{\Big|} \mE\right]\leq
		\EX\left[\sum_{j=0}^\ell  \sum_{i=j}^\ell  \sum_{v\in L_i}  \ap_j(v) \mathrel{\Big|} \mE\right]
	\end{align}
	and we shall be interested in bounding the RHS of the equation.
	Specifically we shall prove that for every $j$,
	\begin{align}\label{eqn:exp_checkp_bounded}
		\EX\left[ \sum_{i=j}^\ell \sum_{v\in L_i}   \ap_j(v) \mathrel{\Big|} \mE\right]\leq \frac{n}{2^j} .
	\end{align}
	We prove this claim  by induction on $j$.
	First, for  $j=0$, it holds that for every $v\in V$, $\ap_0(v)=1$. Hence, for $j=0$,
	\[
	\sum_{i=0}^\ell \sum_{v\in L_i}  \ap_0(v) = \sum_{v \in V} \ap_0(v)=n.
	\]

	Now consider the case $j=1$.
	For every $v\in L_{\geq 1}$, $\ap^1(v)=d(v)/(6\alpha)$ (and this is independent of the event $\mE$).
	Hence,
	\[
	\EX\left[ \sum_{i=1}^\ell \sum_{v\in L_i} \ap_1(v) \mathrel{\Big|} \mE \right]=
	\sum_{i=1}^\ell \sum_{v\in L_i} \EX[\ap_1(v) \mid  \mE]\leq  \sum_{v \in V} d(v)/6\alpha\leq \frac{2m}{6\alpha}\leq n/2,
	\]
	where the last inequality is due to the fact that for every graph with arboricity at most $\alpha$, $m\leq n\alpha$.

	We now assume that the claim holds for every $1\leq j'\leq j-1$, and prove that it holds for $j$.
	By Definition~\ref{def:pcheck}, for $j\geq 2$,
		$\ap_j(v)=  \sum_{u\in \cNei{j-1}(v)}\ap_{j-1}(u).$
	Therefore,
	\begin{align}
		\sum_{i=j}^\ell \sum_{v\in L_i} \ap_{j}(v)=
		\sum_{i=j}^\ell \sum_{v\in L_i}  \sum_{u\in \cNei{j-1}(v)}\ap_{j-1}(u).
	\end{align}
	Recall that by Definition~\ref{dfn:sigma}, for a vertex $u\in L_k$, for every $i\geq k$, $\sigma_i(u)=|\{v \in L_{i} \;:\; u\in S(v)\}|$, and $\sigma(u)=\sum_{i=k}^{\ell}\sigma_i(u)$. Also recall that for  $j\geq 2$ and a vertex $v\in L_i$, $\cNei{j}(v)\subseteq L_{\leq i-1}$.
	Therefore, 
	\begin{align}
		\sum_{i=j}^\ell \sum_{v\in L_i}  \sum_{u\in \cNei{j-1}(v)} \ap_{j-1}(u) &=
		\sum_{i=j}^\ell \sum_{v\in L_i}  \sum_{k=j-1}^{i-1} \sum_{u\in \cNei{j-1}(v)\cap  L_k} \ap_{j-1}(u) \\
		& = \sum_{k=j-1}^\ell \sum_{u\in L_k} \left(\sum_{i=k}^{\ell} \sigma_i(u)\right) \ap_{j-1}(u)
		\\ &=  \sum_{k= j-1}^\ell
		\sum_{u\in L_k}\sigma(u)\cdot \ap_{j-1}(u)
		\numberthis  \label{eqn:sumLsumLk}
	\end{align}
	Hence, we shall bound the expected value of the expression in Equation~\eqref{eqn:sumLsumLk}, conditioned on the event $\mE$.  By Claim~\ref{clm:exp_indep}, Claim~\ref{clm:exp_sigma} and by the induction hypothesis,
	\begin{align}
		\EX\left[\sum_{k= j-1}^\ell
		\sum_{u\in L_k}\sigma(u)\cdot \ap_{j-1}(u) \mathrel{\Big|} \mE\right]&\leq \sum_{k= j-1}^\ell
		\sum_{u\in L_k} \EX[\sigma(u)\mathrel{\Big|} \mE ]\cdot \EX\left[  \ap_{j-1}(u)\mathrel{\Big|} \mE\right]
		\\
		&\leq \sum_{k= j-1}^\ell
		\sum_{u\in L_k}  \frac{1}{2}\cdot \EX\left[  \ap_{j-1}(u)\mathrel{\Big|} \mE\right] \\
		&=\frac{1}{2} \cdot \EX\left[ \sum_{k\geq j-1}^\ell
		\sum_{u\in L_k} \ap_{j-1}(u)\mathrel{\Big|} \mE\right]\\
		&\leq n/2^{j+1}.
	\end{align}
	This completes the proof that for every $j$, Equation~\eqref{eqn:exp_checkp_bounded} holds.
	Summing over all $j's$, we get that
	\[
	\EX\left[\sum_{j=0}^{\ell} \sum_{i=j}^{\ell} \sum_{v\in L_i}\tp_j(v) \mathrel{\Big|} \mE\right]\leq \sum_{j=0}^{\ell} \frac{n}{2^j}\leq 2n.
	\]
	Plugging the above into Equation~\eqref{eqn:expP_vs_checkp} completes the proof.
\end{proof}

\begin{clm}\label{clm:correctness}
	Consider an invocation of \peel$(G,\alpha)$. If $\arb(G)\leq \alpha$, then with probability at least $2/3$, the procedure returns $\YES$. If $\arb(G)> \alphaX$, then with probability at least $1-2/n^4$, the procedure returns $\NO$.
\end{clm}
\begin{proof}
	By Claim~\ref{clm:ub_cost}, for every $v\in V$ and $j\in[\ell]$, the query complexity of \peelv$(v,j)$ is at most $2\tp_{j}(v).$
	Let $\tp(v)=\sum_{j=0}^{\ell}\tp_j(v)$, and for a set $Y$, let $\tp(Y)=\sum_{v\in Y}\tp(v)$.
	The query complexity of \peel$(G,\alpha)$ is bounded by $\sum_{x \in X_0} \sum_{j=0}^{\ell} \tp_j(v)=\sum_{x\in X_0}\tp(v)=\tp(X_0)$.
	We first consider the case that $\arb(G)\leq \alpha$. By Claim~\ref{clm:bound_exp},
	\[
	\EX_{S(V)}\left[\sum_{v\in V} \sum_{j=0}^\ell \tp_j(v) \mid \mE\right]
	\leq 2n\;.
	\]
	Since the different $\tp_j(v)$ values are correlated, we can only use Markov's inequality:
	\begin{align}\label{eqn:markov}
		\Pr\left[\sum_{v\in V}\sum_{j\in[\ell]}\tp_j(v) >20n \mid \mE\right] <\frac{1}{10}.
	\end{align}
	Denote the event that $(\tp(V)\leq 20n \mid \mE)$ by $\calE_2$.
	It holds that
	\[
	\EX_{v\in V}\left[\sum_{j=0}^{\ell}\tp_j(v)\mid \mE\medcap \calE_2\right]\leq \frac{1}{n} \sum_{v\in V} \sum_{j=0}^\ell \tp_j(v) \mid \mE\medcap \calE_2 \leq 20.
	\]
	If follows that, condition on $\mE\medcap \calE_2$,the $\tp(v)$ values are  random variables with expected value at most $20$.
	Hence, $\EX[\tp(X_0)\mid \mE\medcap \calE_2]\leq 20|X_0|$, and by Markov's inequality,
	\begin{align}\label{eqn:markov_qv}
		\Pr[\tp(X_0)>200|X_0| \mid \mE\medcap \calE_2]<\frac{1}{10}\;.
	\end{align}
	Denote the event that $\left(q(X_0) \leq  200\cdot |X_0| \mid \mE\medcap \calE_2\right)$ by $\calE_3$.
	By Claim~\ref{clm:edge-samp}, Equations~\eqref{eqn:markov} and~\eqref{eqn:markov_qv}, and the union bound,
	the event $\mE\medcap \calE_2\medcap \calE_3$ occurs with probability at least $1-\frac{1}{n} -\frac{1}{10}-\frac{1}{10} \geq 2/3$. Therefore, with probability at least $2/3$, the event $\mE\medcap \calE_2\medcap \calE_3$ holds, and by event $\calE_3$,
	$q(X_0)\leq 200t$ so that by Claim~\ref{clm:ub_cost}, $Q$ does no exceed $\UBq$.  In such a case, by Claim~\ref{clm:YES}, since event $\mE$ holds, the algorithm  returns $\YES$ in Step~\ref{step:peel_return}.
	Therefore, with probability at least $2/3$, the procedure returns $\YES$.
	
	We now assume that $\arb(G) > \alphaX$. If the number of allowed queries exceeds $\UBq$, then we are  done.
	Otherwise, by Claim~\ref{clm:edge-samp}, with probability at least $1-1/n^4$, event $\mE$ holds. Condition on $\mE$,  by Claim~\ref{clm:NO}, with probability at least $1-1/n^4$, the procedure returns $\NO$. Hence, the procedure returns $\NO$ with probability at least $1-2/n^4$.
\end{proof}

\subsection{The Search Procedure}\label{sec:search}

In this section we show how, given the procedure \peel, which 
distinguishes between graphs
$G$ for which $\arb(G) \leq \alpha$ and those for which $\arb(G) >\rho \alpha$ for $\rho=100\log^2n$,
we can obtain a $2\rho=200\log^2 n$-factor approximation of $\arb(G)$.

We start with a simple procedure to amplify the success probability of \peel.

\newcommand{\setr}{10\log n}

\algo{
	{\bf Peel-With-Reduced-Error$(G,\alpha)$} \label{reduce}
	\smallskip
	\begin{compactenum}
		\item For $r=1$ to $\setr$  do: \label{step:loop_r}
		\begin{compactenum}
			\item Invoke \peel$(G,\alpha)$, and if it returns $\YES$, then \textbf{Return} $\YES$.
		\end{compactenum}
		\item \textbf{Return} $\NO$.
	\end{compactenum}
}{The procedure \reduce\ is used to amplify the success probability of $\peel$.}{fig:reduce}

\begin{clm}\label{clm:reduce}
	If $\arb(G)\leq \alpha$, then, with probability at least $1-1/n^{3}$, \reduce$(G,\alpha)$ 
returns $\YES.$ If $\arb(G)>\alphaX$, then with probability at least $1-20\log n/n^4$, \reduce$(G,\alpha)$ returns $\NO$. The query complexity of the procedure is $O(n/\alpha)$.
\end{clm}
\begin{proof}
	Assume first that $\arb(G)\leq \alpha$,
	and consider a fixed iteration $r$ of the for loop of \reduce.
	By Claim~\ref{clm:correctness},  \peel$(G,\alpha)$ returns $\YES$ with probability at least $2/3$
	Therefore, the probability that it returns $\NO$ in all $r$ invocations is at most $(1/3)^{\setr}<1/n^{3}$.

	If $\arb(G)>\alphaX$, then by Claim~\ref{clm:correctness}, 	
	every invocation of \peel$(G,\alpha)$, returns $\NO$ with probability at least $1-2/n^4$. Hence, the probability that the procedure returns $\NO$ in all $r$ invocations is at least $1-2r/n^4>1-20\log n/n^4$.
	
	Finally, since every invocation of \peel\ does not exceed $\UBq$ queries for $t=\setT$, and \reduce\ makes at most $\setr$ calls to \peel, the query complexity is $O(n/\alpha)$, as claimed.
\end{proof}

\algo{	
	{\bf Estimate-Arboricity$(G)$} \label{est}
	\smallskip
	\begin{compactenum}
		\item Set $\ta=n$.
		\item While $\ta >1$ do:
		\begin{compactenum}
			\item Invoke \reduce$(G,\alpha)$.  If the algorithm returns $\NO$, then \textbf{return}
                $\hat{\alpha}=  
                     \ta$. Otherwise,
			let $\ta=\ta/2$.
		\end{compactenum}
		\item Return $\hat{\alpha}=1$.
	\end{compactenum}
}
{The procedure gets query access to a graph $G$, and returns an estimate of $\arb(G)$.}{fig:est}

We are finally ready to prove our main theorem, which we restate here for the sake of convenience.

\upperBound*
\begin{proof}
	By Claim~\ref{clm:reduce},
	every invocation of  $\reduce(G,\ta)$ with a value $\ta$ such that $\arb(G)\leq \ta$ returns $\YES$  with probability at least $1-1/n^{3}$. Since there are at most $\log (n/\arb(G)) \leq \log n$ iterations with such $\ta$ values, by the union bound, with probability at least $1-\log n/n^{3}$, all such invocations will return $\YES$.
Next, for any invocation of $\reduce(G,\ta)$ with $\arb(G)/(100\log^2 n)\leq \ta < \arb(G)$,
the procedure may return either $\YES$ or $\NO$.
Once $\ta$ goes below $\arb(G)/100\log^2 n$ (so that $\arb(G) > 100\log^2 n\ta$),
  by Claim~\ref{clm:reduce}, $\reduce(G,\ta)$  returns $\YES$  with probability at least $1-20\log n/n^4$.
Recall that the value of $\ta$ is decreased by a multiplicative factor of $2$ in each iteration of algorithm, and the algorithm returns  $\widehat{\alpha}=\ta$ for the first (largest) value of $\ta$ such that $\reduce(G,\ta)$  returns $\NO$.
It follows that with probability $1-\log n/n^3-20\log n/n^4 > 1-O(1/n^2)$  the algorithm makes $O(\log n)$ invocations to $\reduce(G,\ta)$, all with $\ta \geq  \arb(G)/(200\log^2 n)$,
and returns a value $\widehat{\alpha}$ that satisfies
\[
\arb(G)/(200\log^2 n)\leq \widehat{\alpha} \leq \arb(G)\;.
\]
	
	By Claim~\ref{clm:reduce}, the query complexity  of \reduce$(G,\ta)$ for values $\ta \geq  \arb(G)/(200\log^2 n)$
is $O(n/\ta)=O(n\log^2 n/\arb(G))$.
	Hence, with probability at least $1-O(1/n^2)$, the query complexity is
$O(n\log^3 n/\arb(G))$.

	If $\ta$ reaches values smaller than $\arb(G)/(200\log^2 n)$, then we can bound the query complexity and running time by
	$O(n\cdot \ta(G))$.
	This is true since if the number of  queries exceeds $O(n\ta)$ then the algorithm may abort, as  it implies that $\ta$ is too small (since the number of queries is always bounded by $2m=O(n\arb(G))$).
	If the algorithm aborts then it outputs $\widehat{\alpha}=1$.	
	Also, there are at most $\log(\arb(G))=O(\log n)$ iterations with   values $\ta\leq \arb(G)$.
	Therefore, the expected query complexity of \est$(G)$ is
	$O(n\log^3 n/\arb(G)+(1/n^2)\cdot (n\cdot \arb(G))\cdot \log n))=O(n\log^3 n/\arb(G))$.
\end{proof}

Finally, for the sake of completeness, we prove the proposition regarding the lower bound on  any algorithm for approximating the arboricity.

\begin{proof}[Proof of Proposition~\ref{prop:lb}]
	Consider the following two families of graphs, where within each family the graphs only differ by the labels of the vertices and edges. In the first family, there is a clique of size $\alpha$ and the rest of the vertices are isolated. The graphs of the second family are identical, except that  the clique is of size $\alpha\cdot 2k$. Any algorithm that returns a $k$-multiplicative 
approximation of $\arb(G)$ with probability at least $2/3$ must be able to distinguish between these two families. Since the probability of hitting a clique vertex is $O(\alpha k/ n)$, a lower bound of $\Omega(n/(\alpha k))=\Omega(n/(\arb(G)\cdot k))$ follows.
\end{proof}

\section{Adaptation to the 
         Streaming Model}\label{sec:stream}

Our algorithm can be  adapted to the streaming model using $O(\log n)$ passes. 
In general, it is known that  any sublinear-time algorithm in the incidence list  query model with ``adaptivity depth $k$" (see definition below), can be implemented in the streaming model with $2k$ passes. The reason is that all types of queries in the incidence list query model, can be computed using a single pass over the stream: degree queries can be computed using a simple counter, and neighbor queries can be simulated using $\ell_0$ samplers (e.g., that of~\cite{jowhari2011tight}).
Furthermore, the space requirement of the streaming variant can be directly  bounded by the running time of the simulated sublinear algorithm, up to $\poly(\log n)$ factors resulting from the $\ell_0$ samplers.

\begin{dfn}[Depth of adaptivity]
	We say that an algorithm $\mA$  in the incidence list model \textsf{has adaptivity depth} $k$ if the following holds.
	For every execution of $\mA$, the set of queries $\mQ$ performed by it can be partition into $k$ sets $\mQ_{1}, \ldots, \mQ_{k}$, so that for every $j\in [1,k]$,
	 the set of queries $\mQ_j$ can be computed based solely on responses to queries $\mQ_{0} \cup \ldots \cup \mQ_{j-1}$.
\end{dfn}

For a thorough investigation of an adaptivity hierarchy in property testing see~\cite{CanonneGur18}.

We claim that algorithm \peel\ has depth of adaptivity $O(\log n)$.
\begin{clm}
	The adaptivity depth of Algorithm \peel\ is $O(\log n)$.
\end{clm}
\begin{proof}
	We describe how to partition the algorithm's queries to $O(\log n)$ sets, where each set only depends on previous ones:
	Let $\mQ_1$ be the set of degree queries performed on the vertices of $X_0$ (in order to determine \peel$(v,0)$ for every $v\in X_0$). 	
For every $v$ and $j\in [1,\ell]$, let $\mQ_{v,2j}$ be the set of neighbor queries performed during the invocation of  \peel$(v,j)$, and let $\mQ_{v,2j+1}$
be the set of degree queries performed 
during the invocation of  \peel$(v,j)$. 
Let 	 $\mQ_{2j}$ and $\mQ_{2j+1}$ be the set of neighbor and pair queries, respectively, required to compute \peel$(v,j)$ on all vertices in $X_j$, so that   $\mQ_{2j}=\{Q_{v,2j}\}_{v\in X_j}$ and
$\mQ_{2j+1}=\{Q_{v,2j+1}\}_{v\in X_j}$.
By the design of our algorithm, it holds that for every $v$, $Q_{v,2j}$ is determined by the queries and their  responses   $\{Q_{v,j'}\}_{j'\leq 2j-1}$: indeed observe that already at the end of the invocation of $\peelv(v,j-1)$, we know the set of vertices from which we should perform neighbor queries in the case $\peel(v,j)$ will be invoked (this is the reason we can compute $q_{j}(v)$ already at the end of \peel$(v,j-1)$). The set of degree queries $\mQ_{v,2j+1}$ is then determined by the identity of the neighbors that are returned as answers to the set of queries $\mQ_{v,2j}$. Hence, for every $j'\in[1,2\log n+2]$, $\mQ_{j'}$	only 
depends on the queries and responses to queries in $\mQ_{\leq j'}$, and it follows that the adaptivity depth of \peel\ is $2\log n+2=O(\log n)$.
\end{proof}

We continue to describe how to adapt our algorithm to the streaming setting.

\vspace{-1em}
\paragraph{Adapting the procedures \textsf{Peel-With-Reduced-Error} and \textsf{Estimate-Arboricity}.}
We  modify the procedure \reduce\ so that, given $G$ and $\alpha$,  the  $r=10\log n$ invocations of $\peel(G,\alpha)$ will take place in parallel, rather than sequentially. Thus, for every value of $\alpha$, there are $10\log n$ invocations of $\peel(G,\alpha)$ in parallel. If any of these invocations returns $\YES$, the (modified) version of  \reduce\ returns YES, and otherwise it returns no.

We  continue to explain how to adapt \est.
Given a lower bound $\alpha$ on $\arb(G)$,
in order to get an estimate of $\arb(G)$, we proceed as follows. We  invoke \peel$(G,\ta)$ with  guesses $\ta=\alpha, 2\alpha, ..., n$ in parallel, and return $\widehat{\alpha}=2\ta$ for  the smallest value $\ta$ for which \reduce$(G,\ta)$ returns $\YES.$ An almost identical analysis to that of Theorem~\ref{thm:upperBound}, proves that with high probability, the returned value is an $O(\log^2)$-approximation of $\arb(G)$.
Therefore, we obtain the following result.
\ubStreaming*

\paragraph{Comparison to existing streaming results.}
We would like to compare Theorem~\ref{thm:upperBoundStreaming}, to the existing streaming algorithms for approximating the densest subgraph (as they can easily be altered to approximate the arboricity).
Recall that our algorithm works by iteratively considering increasingly smaller guesses  of the value of $\arb(G)$, starting from $\ta=n$, and halving the guess at each step, where the minimum possible value of $\arb(G)$ is  the average degree of $G$, $d_{avg}(G)$.
A common challenge to one-pass streaming algorithms is that this ``search process'' must be done simultaneously for all possible guesses of $\arb(G)$ during the single pass over the stream, resulting in a space complexity of $O(m/d_{avg})=O(n)$.
Unfortunately, this state of affairs obscures the true dependence on the parameter that the algorithm is trying to approximate. To address this challenge, various works on graph parameter estimation in the streaming model assume that they are given a rough estimate on the parameter at question, and the goal is to achieve a more accurate one (see, e.g.,~\cite{mcgregor, KP17, SV20}).  Indeed if such a lower bound $\alpha$ on $\arb(G)$ is given to the algorithms of~\cite{BHNT15} and~\cite{mcgregor2015densest}, then their space complexity is reduced to $\widetilde{O}(m/\alpha)$.

Therefore, we compare our streaming variant with the state of the art streaming algorithms, under the assumption that a lower bound $\alpha$ on $\arb(G)$ is given as input.
In such setting our algorithm is an $O(\log n)$-passes, $\widetilde{O}(n/\alpha)$-space, $O(\log^2 n)$-approximation algorithm, where the state of the art by~\cite{mcgregor2015densest} is a $1$-pass, $\widetilde{O}(m/\alpha)$, $(1+\eps)$-approximation algorithm.
Hence, our algorithm improves on the space complexity by  factor of $d_{avg}$, at the cost of performing $O(\log n)$ passes over the stream, and an $O(\log^2 n)$-approximation factor.

	\bibliographystyle{plain}
	\bibliography{est_arb_bib}
	
\end{document}